\documentclass[12pt, draftclsnofoot, onecolumn]{IEEEtran}
\usepackage{algorithm,algorithmic,amsbsy,amsmath,amssymb,epsfig,mathrsfs,url,color,cite}
\usepackage{subcaption} 
\renewcommand{\baselinestretch}{1.5}
\usepackage[letterpaper,margin=1.95cm,footskip=0.9cm]{geometry}

\newcommand{\beq}{\begin{equation}}
\newcommand{\eeq}{\end{equation}}
\newcommand{\beqn}{\begin{eqnarray}}
\newcommand{\eeqn}{\end{eqnarray}}

\newtheorem{theorem}{\textbf{Theorem}}
\newtheorem{lemma}{\textbf{Lemma}}
\newtheorem{proposition}{\textbf{Proposition}}
\newtheorem{corollary}{\textbf{Corollary}}

\newenvironment{proof}[1][Proof]{\begin{trivlist}
\item[\hskip \labelsep {\bfseries #1}]}{\end{trivlist}}

\newcommand{\qed}{\nobreak \ifvmode \relax \else
      \ifdim\lastskip<1.5em \hskip-\lastskip
      \hskip1.5em plus0em minus0.5em \fi \nobreak
      \vrule height0.55em width0.5em depth0.25em\fi}

\usepackage{geometry}
\usepackage{enumerate} 
\usepackage{fancyhdr} 
\usepackage[labelformat=simple]{subcaption}
\renewcommand\thesubfigure{(\alph{subfigure})}

\long\def\symbolfootnote[#1]#2{\begingroup%
\def\thefootnote{\fnsymbol{footnote}}\footnote[#1]{#2}\endgroup}

\newcount\colveccount
\newcommand*\colvec[1]{
        \global\colveccount#1
        \begin{pmatrix}
        \colvecnext
}
\def\colvecnext#1{
        #1
        \global\advance\colveccount-1
        \ifnum\colveccount>0
                \\
                \expandafter\colvecnext
        \else
                \end{pmatrix}
        \fi
}

\begin{document}
      
\title{Unified Stochastic Geometry Model for MIMO Cellular Networks with Retransmissions }
\author{\IEEEauthorblockN{Laila~Hesham~Afify, Hesham ElSawy, Tareq Y. Al-Naffouri, and Mohamed-Slim Alouini\\} 
\vspace{0.3cm}
\IEEEauthorblockA{ 
King Abdullah University of Science and Technology (KAUST) \\
Email: \{laila.afify, hesham.elsawy, tareq.alnaffouri, slim.alouini\}@kaust.edu.sa}} 

\maketitle

\begin{abstract}

This paper presents a unified mathematical paradigm, based on stochastic geometry, for downlink cellular networks with multiple-input-multiple-output (MIMO) base stations (BSs). The developed paradigm accounts for signal retransmission upon decoding errors, in which the temporal correlation among the signal-to-interference-plus-noise-ratio (SINR) of the original and retransmitted signals is captured. In addition to modeling the effect of retransmission on the network performance, the developed mathematical model presents twofold analysis unification for MIMO cellular networks literature. First, it integrates the tangible decoding error probability and the abstracted (i.e., modulation scheme and receiver type agnostic) outage probability analysis, which are largely disjoint in the literature. Second, it unifies the analysis for different MIMO configurations. The unified MIMO analysis is achieved by abstracting unnecessary information conveyed within the interfering signals by Gaussian signaling approximation along with an equivalent SISO representation for the per-data stream SINR in MIMO cellular networks. We show that the proposed unification simplifies the analysis without sacrificing the model accuracy. To this end, we discuss the diversity-multiplexing tradeoff imposed by different MIMO schemes and shed light on the diversity loss due to the temporal correlation among the SINRs of the original and retransmitted signals. Finally, several design insights are highlighted. 

\end{abstract}


\section{Introduction}

 Multiple-input-multiple-output (MIMO)  transmission offers diverse options for antenna configurations that can lead to different diversity and multiplexing tradeoffs, which can be exploited to improve several aspects in wireless networks performance. For instance, link capacity gains can be harvested by multiplexing several data streams into the same channel via MIMO spatial multiplexing. Enhanced link reliability can be obtained by transmit and/or receive diversity. The network capacity can be improved by accommodating more users equipment (UEs) per channel via multi-user MIMO techniques. Last but not least, enhanced interference management can be achieved via beamforming or interference alignment techniques such that dominant interference sources are suppressed at the receivers side, and hence, the signal-to-interference-plus-noise ratio (SINR) is improved. 


Motivated by its potential gains, MIMO is considered an essential ingredient in modern cellular networks and 3GPP standards to cope with the ever-growing capacity demands. However, the MIMO operation is understood and its associated gains are quantified for elementary network {settings~\cite{paulraj_book, mimo_comm_book, larsson_book}}, which do not directly generalize to cellular networks. The operation of large-scale cellular network is highly affected by inter-cell interference, which emerges from spatial frequency reuse. Therefore, to characterize MIMO operation and quantify its potential gains in cellular network, the impact of per-base station (BS) precoding on the aggregate interference as well as the effect of the aggregate interference on the received SINR after MIMO post-processing should be characterized. {Exploiting recent advances in stochastic geometry analysis, several mathematical frameworks are developed to study MIMO operation in cellular networks~\cite{ STBC_harpret, ralph2, Int_align_tony, Massive_harpret, MIMO_heath ,CoMP_sakr, CoMP_martin, coordination1_martin, Spatiotemporal,  ralph_coop, Mimo_ordering,Mimo_load_harpreet, eid_Mimo, error_per2, asymptotic_SE}. Stochastic geometry does not only provide systematic and tractable framework to model MIMO operation in interference environments, it  also captures the behavior of realistic cellular networks as reported in~\cite{tractable_app, martin_ppp, survey_h}.} The authors in \cite{STBC_harpret} study the SINR coverage probability of orthogonal space-time block codes (OSTBC). Studies for the outage probability and ergodic rate for space-division-mutiple-access (SDMA) MIMO, also known as multi-user MIMO, are available in \cite{Mimo_ordering,Mimo_load_harpreet}. Coverage probability improvement via maximum-ratio-combining (MRC) with spatial interference correlation is quantified in \cite{ralph2}. The potential gains of beamforming and interference alignment in terms of SINR coverage and network throughput are quantified in \cite{Int_align_tony, Massive_harpret, MIMO_heath}.  
 Network MIMO via BS cooperation performance in terms of outage probability and ergodic rate are studied in \cite{CoMP_sakr, CoMP_martin, coordination1_martin, Spatiotemporal, ralph_coop}. Average symbol error probability (ASEP) and average pairwise error probability (APEP) for several MIMO configurations are studied in \cite{eid_Mimo, error_per2}. Asymptotic analysis for minimum mean square error MIMO receivers is conducted in \cite{asymptotic_SE}. 


Despite that the mathematical models presented in \cite{ralph2, Spatiotemporal, STBC_harpret, Int_align_tony, Massive_harpret, MIMO_heath ,CoMP_sakr, CoMP_martin, coordination1_martin, Mimo_ordering,Mimo_load_harpreet, eid_Mimo, error_per2, asymptotic_SE, ralph_coop} are all based on stochastic geometry, there are significant differences in terms of the analysis steps as well as the level of details provided by each model. The majority of the models focus on the outage probability and ergodic capacity for simplicity~\cite{ralph_coop, ralph2, STBC_harpret, Int_align_tony, Massive_harpret, MIMO_heath ,CoMP_sakr, CoMP_martin, coordination1_martin, Mimo_ordering,Mimo_load_harpreet, Spatiotemporal}. While both outage probability and ergodic rate are fundamental key performance indicators (KPIs) in wireless communication, they convey no information about the underlying modulation scheme, constellation size, or receiver type. Considering more tangible KPIs, such as decoding error probability and average throughput\footnote{Throughput is defined as the number of successfully transmitted bits per channel use.}, requires alternative and more involved analysis as shown in \cite{ eid_Mimo, error_per2}. {The decoding error probability analysis in \cite{ eid_Mimo, error_per2} is different and more involved when compared to the outage analysis in \cite{ralph_coop, ralph2, STBC_harpret, Int_align_tony, Massive_harpret, MIMO_heath ,CoMP_sakr, CoMP_martin, coordination1_martin, Mimo_ordering,Mimo_load_harpreet, Spatiotemporal}. The complexity of the error probability analysis can be attributed to the necessity to model the aggregate interference signal at the complex baseband level and statistically account for the transmitted symbols from each interferer~\cite{moe_win, moe_win_error}.} Furthermore, the analysis in each of the mathematical  models presented in \cite{ralph2, ralph_coop, STBC_harpret, Int_align_tony, Massive_harpret, MIMO_heath ,CoMP_sakr, CoMP_martin, coordination1_martin, Mimo_ordering,Mimo_load_harpreet, eid_Mimo, error_per2, Spatiotemporal} is highly dependent on the considered MIMO configuration. Having mathematical models that are significantly different in the analysis steps from one KPI to another and from one antenna configuration to another makes it challenging to conduct comprehensive studies that include different antenna configurations and compare their performances in terms of different KPIs. Furthermore, the effect of retransmission upon decoding errors on the network performance is only studied for cooperating single-antenna BSs~\cite{Spatiotemporal}, which is different from the case when both BSs and UEs are equipped with multiple antennas and precoding/combining is applied.  

 

This paper presents a unified mathematical paradigm, based on stochastic geometry, to study the average error probability, outage probability, and ergodic rate for cellular networks with different MIMO settings{\footnote{{This paper is an extension of the work reported in \cite{mimo_icc2016} to be presented in \emph{IEEE International Conference on Communications (ICC)}, 2016. }}. In addition to unifying the analysis for different KPIs and different MIMO configurations, the presented model accounts for the effect of data retransmission upon decoding errors, in which the temporal correlation between the original and retransmission SINRs is captured. The developed framework relies on abstracting the unnecessary information conveyed within the interfering symbols by Gaussian signals along with a unified equivalent SISO-SINR representation for the different MIMO schemes. The Gaussian signaling approximation simplifies the decoding error analysis without sacrificing the model accuracy and aligns the decoding error analysis with the outage probability analysis. The Gaussian signaling abstraction is proposed in \cite{Laila_letter} for a SISO cellular network, in which we generalize to MIMO cellular networks in this paper.

{It is worth noting that the SISO-SINR representation used in this paper was used to study OSTBC in \cite{STBC_harpret, andrews_adhoc, effect_coch_SDT}, receive  diversity  in \cite{ralph2}, and MU-MIMO in \cite{Mimo_ordering}.} However, the analysis in \cite{ralph2, Mimo_ordering, STBC_harpret,andrews_adhoc, effect_coch_SDT} is confined to outage probability and ergodic rate only. Also, the OSTBC studies in \cite{STBC_harpret,andrews_adhoc, effect_coch_SDT} used the SISO-SINR model as an approximation for the outage probblity characterization. In this paper, we extend the equivalent SISO-SINR model to evaluate decoding error probability, outage probability, ergodic rate, and throughput in MIMO cellular networks with transmit diversity (MISO), receive diversity (SIMO), OSTBC, MU-MIMO, and spatial multiplexing with zero-forcing receivers (ZF-Rx). Furthermore, we prove that the equivalent SISO-SINR model is \emph{exact} in the case of OSTBC MIMO scheme.  


 To the best of our knowledge, this paper is the first to present such unified study that can be used to investigate the diversity-multiplexing tradeoffs between different MIMO configurations in terms of different performance metrics and provide guidelines for MIMO cellular network design. Particularly, we propose an automated reliable strategy to determine which MIMO settings to deploy in order to achieve a desired network objective under performance constraints. The main contributions of the developed framework can be summarized in the following points:
\begin{itemize}
\item Developing a unified mathematical model that bridges the gap between error probability, outage probability, and ergodic rate analysis. Hence, it is possible to look at all three performance metrics within a single study.
\item Extending the SISO-SINR model to evaluate outage probability and ergodic rate for ZF-Rx MIMO, showing that the equivalent SISO-SINR is exact in the case of OSTBC MIMO (i.e., not an approximation as assumed in \cite{STBC_harpret,andrews_adhoc, effect_coch_SDT}), and using the equivalent SISO-SINR to evaluate decoding error probability for all considered MIMO schemes. Note that the decoding error probability model presented in \cite{eid_Mimo}  differs across the MIMO configurations due to the different effect of the precoding/combining matrices on the aggregate interference when accounting for modulation type and constellation size. 
\item Bypassing the complex baseband interference analysis, which simplifies the decoding error probability analysis and reduces the computational complexity of the final expressions,  when compared to \cite{eid_Mimo}, without compromising the model accuracy. 
\item Accounting for the signal retransmission upon decoding failure in a MIMO framework, in which the correlation among the original and retransmitted SINRs is captured. 
\item Revealing the cost of multiplexing, in terms of  outage probability and decoding error, in large-scale cellular networks. We also show the appropriate diversity compensation for such cost. 
\end{itemize}

\subsection{Organization \&  Notation}

{The paper is organized as follows. Section \ref{sec:system_model} presents the generic system model for a downlink cellular network deploying an arbitrary MIMO setup. In Section \ref{sec:gauss_apprx} the Gaussian signaling approximation and the equivalent SISO-SINR model are presented. The unified model with and without retransmission is presented in Section \ref{unified_ana}. Section \ref{sec:mimo_schemes} illustrates how to represent different MIMO schemes via the equivalent SISO-SINR model. The model validation, via Monte-Carlo simulations,  and the key findings of the paper are presented in Section \ref{sec:sim_results} and the paper is concluded in Section \ref{sec:conc}.} Throughout the paper, we use the following notations:  small-case bold-face letters $(\textbf{x})$ denote column vectors, upper-case bold-face letters $(\textbf{X})$ denote matrices,  $\left( \cdot\right)^T$ and $\left( \cdot\right)^H$ denote the transnpose and conjugate operators, respectively.
$\| \cdot \|$ is the Euclidean norm operator, $\mathbb{E}_x\left[ \cdot \right]$ and $\text{Var}_{x}\left[ \cdot \right]$ are the expectation and variance computed with respect to the random variable $x$, respectively, and $\text{erfc}(x)=\frac{2}{\sqrt{\pi}} \int_x^{\infty} e^{-t^2} \mathrm{d}t$ is the complementary error function.












\section{System Model}

\label{sec:system_model}
\subsection{Network and Propagation Models}

We consider a single-tier downlink cellular network, where the BSs locations are modeled by a homogeneous PPP $\Psi_B$ with intensity $\lambda_B$. UEs are distributed according to an independent homogeneous PPP $\Psi_u$ with intensity $\lambda_u$. BSs and UEs are equipped by $N_t$ and $N_r$ colocated antennas,\footnote{Colocated antennas is a common assumption in MIMO models based on stochastic geometry analysis to maintain the model tractability  \cite{ralph2, STBC_harpret, Int_align_tony, Massive_harpret, MIMO_heath, Mimo_ordering,Mimo_load_harpreet, eid_Mimo, error_per2, asymptotic_SE}.} respectively. Conditions on the relation between $N_r$ and $N_t$ depend on the MIMO setup under study, as will be shown later. Without loss of generality, we assume that $\Psi_B = \{r_o, r_1, r_2, \cdots \}$ contains the ascending ordered distances of the BSs from the origin (i.e., $r_o<r_1<r_2 $) and that the analysis is conducted on a test user located at the origin \cite{moe_win}. According to Slivinyak's Theorem, there is no loss of generality in this assumption \cite{martin_book}. Assuming nearest BS association, the test user is subject to interference from the BSs in $\Psi^o = \Psi_B \setminus r_o$, in which the distance between the test user and its serving BS has the probability distribution function (PDF) $f_{r_o}\left( r \right)= 2 \pi \lambda_B r e^{-\pi \lambda_B r^2}$, $r_o>0$. Let $p$ be the independent transmission probability for each BS in $\Psi^o$, then, the point process of the active interfering BSs $\tilde{\Psi}^o \subseteq  {\Psi}^o$ after independent thinning is also a PPP but with intensity $\lambda=p \lambda_B$\cite{martin_book}.  This assumption is used to reflect load awareness and/or frequency reuse as discussed in \cite{eid_Mimo, load_aware_harpreet}. Note that $p$ can be calculated as in \cite{on_cog}, and setting $p=1$ gives the traditional saturation condition (i.e., $\lambda_u \gg \lambda$) where all BSs are active.

A distance-dependent power-law path-loss attenuation is employed, in which the signal power attenuates at the rate $r^{-\eta}$ with the distance $r$, where $\eta > 2$ is the path-loss exponent. In addition to path-loss attenuation, we consider a Rayleigh multi-path fading environment between transmitting and receiving antennas. That is, the channel gain matrix from a transmitting BS to a generic UE, denoted as $\textbf{H} \in \mathbb{C}^{N_r\times N_t}$, has independent zero-mean unit variance complex Gaussian entries, such that  $\textbf{H} \sim \mathcal{CN}\left(0,\textbf{I} \right)$, where $\textbf{I} $ is the identity matrix. 


\subsection{Downlink MIMO Received Signal Model}
\label{sec:rxed_signal}

 For a general MIMO setup in Rayleigh fading environment, and considering arbitrary precoding/combining schemes, the complex baseband received signal vector is expressed as

\small
\begin{align}
\label{eq:mimo_rxed_signal1}
\textbf{y}&= \sqrt{\frac{P}{ r_o^{{\eta}}}} \textbf{W}_o  \textbf{H}_o \textbf{V}_o  \textbf{s} +\sum_{r_i \in \tilde{\Psi}^o  } \sqrt{\frac{P}{r_i^{{\eta}}}}  \textbf{W}_o \textbf{H}_i \textbf{V}_i \textbf{s}_i   +\textbf{W}_o  \textbf{n},
\end{align}
\normalsize

\noindent where {$P=\frac{E_s}{N_t}$} is the transmit power per antenna at the BSs such that $E_s$ is the energy per symbol, $\textbf{H}_{o} \in \mathbb{C}^{N_r \times N_t}$ is the useful channel matrix from the serving BS, and  $\textbf{H}_i  \in \mathbb{C}^{N_r \times N_t}$ is the interfering channel matrix from the $i^{th}$ interfering BS, $\textbf{H}_o$ and $\textbf{H}_i$, have i.i.d  $\mathcal{CN}\left( 0,1\right)$ entries, $\textbf{s} \in \mathbb{C}^{L \times 1}$ and $\textbf{s}_i  \in \mathbb{C}^{L \times 1}$ are, respectively, the intended and interfering symbols vector, where $L$ represents the number of multiplexed data streams\footnote{According to the employed MIMO setup, we might need to introduce a slight abuse of notation to preserve the convention used in \eqref{eq:mimo_rxed_signal1}. For instance, in multi-user MIMO setup with $\mathcal{K}$ single-antenna users, the parameter $N_r$ should be replaced with $\mathcal{K}$ so that the signal model in \eqref{eq:mimo_rxed_signal1} remains valid.}. The symbols in  $\textbf{s}$  and $\textbf{s}_i$ are independently drawn from an equiprobable two dimensional unit-energy constellations. The matrices  $\textbf{V}_o  \in \mathbb{C}^{N_t \times L}$ and $\textbf{V}_i \in \mathbb{C}^{N_t \times L}$ are the intended and interfering precoding matrices at the intended BS and the $i^{th}$ interfering BS, respectively, while $\textbf{W}_o  \in \mathbb{C}^{L \times N_r}$ is the combining matrix at the test receiver. Note that  $\textbf{W}_o$,  $\textbf{V}_o$, and  $\textbf{V}_i$  are determined based on the employed MIMO scheme. $\textbf{n} \in \mathbb{C}^{N_r \times 1}$ is the zero-mean additive white Gaussian noise vector with covariance matrix $\mathcal{N}_o \textbf{I}_{N_r }$, where $\textbf{I}_{N_r }$ is the identity matrix of size $N_r$.  Last, we assume a per-symbol maximum likelihood (ML) receiver at the test UE to decode the symbols in $\textbf{s}$. 

\section{Gaussian Signaling Approximation \& Equivalent SISO  Representation}
\label{sec:gauss_apprx}

Assuming per-stream symbol-by-symbol ML receiver, the employed precoding, combining, and equalization techniques decouple symbols belonging to different streams (i.e., in case of multiplexing) and/or combine symbols belonging to the same stream (i.e., in case of diversity) at the decoder to allow disjoint and independent symbol detection across the multiplexed data streams. Hence, the precoding and combining matrices ($\textbf{W}_o$ and $\textbf{V}_o$) are tailored to $\textbf{H}_o$ such that the product $\textbf{W}_o \textbf{H}_o \textbf{V}_o$ gives the appropriately scaled identity matrix of size $L$.  Without loss of generality, let us focus on the decoding performance of a generic symbol in the $l^{th}$ stream, in which the instantaneous received signal after applying combining/precoding techniques is given by

\scriptsize
\begin{align}
\label{eq:mimo_rxed_signal}
{{y}}_{l}&= \sqrt{\frac{{P}}{ r_o^{{\eta}}}} \underset{g_o s_l}{\underbrace{\sum_{k=1}^L \overset{[0 \; 0 ... g_o ... 0 \; 0]}{\overbrace{\bar{\textbf{w}}^T_{o,l}  \textbf{H}_o {\textbf{v}}_{o,k}}} s_{k}}} +\sum_{r_i \in \tilde{\Psi}^o  } \underset{I_i}{\underbrace{ \sqrt{\frac{P}{ r_i^{{\eta}}}} \sum_{k=1}^L \overset{{a}^{(i)}_{l,k}}{\overbrace{\bar{\textbf{w}}^T_{o,l}  \textbf{H}_i {\textbf{v}}_{i,k}}}  s_{i,k}}} +   \bar{\textbf{w}}^T_{o,l}   \textbf{n},
\end{align}
\normalsize

\noindent such that $l \in \left\{1, \cdots, L\right\}$, $\bar{\textbf{w}}_{o,l}$ is the $l^{th}$ column of matrix $\bar{\textbf{W}}_o={\textbf{W}_o^T}$, and ${\textbf{v}}_{i,k}$ is the $k^{th}$ column of matrix $\textbf{V}_i$. Further, $g_o$ is a real random scaling factor that appears in the intended signal due to the equalization applied for detecting the desired symbol, $a^{(i)}_{l,k}$, $\forall k$ are the complex random coefficients combining the interfering symbols from the antennas of the $i^{th}$ interfering BS. As shown in \eqref{eq:mimo_rxed_signal}, the coefficients $a^{(i)}_{l,k}$ are generated from the product $\bar{\textbf{w}}^T_{o,l} \textbf{H}_i \textbf{v}_{i,k}$, which capture per-BS precoding and test receiver combining effects on the aggregate interference. Since $\bar{\textbf{W}}_{o} $ and $\textbf{V}_{i} $ are designed independently from each other and from $\textbf{H}_i $, 
the coefficients $a^{(i)}_{l,k}$ are not controllable and depend on  the random values in $\textbf{H}_o $, $\textbf{H}_i $, and the channel matrix between the $i^{th}$ interfering BS and its associated users (denoted as $\tilde{\textbf{H}}_i $). 

It is clear, from \eqref{eq:mimo_rxed_signal}, that the aggregate interference seen at the decoder of the test UE is highly affected by the per-interfering BS precoding scheme ($\textbf{v}_{i,k}$), the number of streams transmitted by each interfering BS ($L$), the per-stream transmitted symbol ($s_{i,k}$), and the employed combining technique ($\bar{\textbf{w}}_{o,l}$) at the test UE. Therefore, characterizing the aggregate interference in \eqref{eq:mimo_rxed_signal} is essential to characterize and quantify MIMO operation in cellular networks. The aggregate interference term contains three main sources of randomness, namely, the network geometry, the channel gains\footnote{The precoding and combining matrices are functions of the channel gains.}, and the interfering symbols. Treating interference as noise, the average decoding performance of a symbol conveyed in a signal in the form of \eqref{eq:mimo_rxed_signal} is characterized through the following two steps:
\begin{enumerate}[(i)]
\item \emph{conditionally} averaging over the transmitted symbols and noise (while {conditioning} on the network geometry and channel gains),
\item then averaging over the channel gains and network geometry,
\end{enumerate}
The averaging in step (i) is done based on the modulation scheme, constellation size, and receiver type. Then, in step (ii), the averaging is done using stochastic geometry analysis.


 The average decoding error performance in step (i) is only characterized for certain distributions of additive noise channels (e.g., Gaussian~\cite{slim_comm_book}, Laplacian~\cite{Laplace1,Laplace2}, and Generalized Gaussian~\cite{Generalized_1}), in which the additive white Gaussian noise (AWGN) channel represents the simplest case. Accounting for the exact distribution of the interfering symbols, the interference-plus-noise distribution does not directly fit into any of the distributions where the average decoding error performance is known. Hence, the averaging step (i) cannot be directly conducted unless the interference-plus-noise term is expressed or approximated via one of the distributions where the average decoding error performance is known. The authors in \cite{eid_Mimo} proposed the {\em equivalent-in-distribution} (EiD) approach where an exact conditional Gaussian representation for the aggregate interference in  \eqref{eq:mimo_rxed_signal} is  achieved. Hence, the conditional error probability analysis (i.e., step (i)) is conducted via error probability expressions for AWGN channels, followed by the deconditioning step in (ii). The main drawback of the EiD approach is that it requires characterizing the interference signals at the baseband level to achieve the conditional Gaussian representation, which complicated both averaging steps (i), (ii) specially in MIMO networks. Furthermore, the EiD approach for the error probability analysis in \cite{eid_Mimo} is disjoint from the outage probability and ergodic rate analysis in \cite{STBC_harpret, Int_align_tony, Massive_harpret, MIMO_heath ,CoMP_sakr, CoMP_martin, coordination1_martin, Mimo_ordering,Mimo_load_harpreet, Spatiotemporal}.

To facilitate the error probability analysis and achieve a unified error probability, outage probability and ergodic rate analysis, we only account for the entries in the intended symbol vector $\textbf{s}$ of \eqref{eq:mimo_rxed_signal1} and abstract the entries in $\textbf{s}_i$ by i.i.d. zero-mean Gaussian signals $\tilde{\textbf{s}}_i$  with unit-variance. Such abstraction ignores the unnecessary and usually unavailable information of the interfering signals. Assuming Gaussian signaling for the interfering symbols (i.e., the entries in $\tilde{\textbf{s}}_i$ are Gaussian), \eqref{eq:mimo_rxed_signal} can be rewritten as
\begin{align}\label{eq:mimo_rxed_signal2}
{{y}}_{l}&= \sqrt{\frac{P}{r_o^{{\eta}}}}\:  {g_o}  s_{l} +\sum_{r_i \in \tilde{\Psi}^o  } \underset{I_{i}}{ \underbrace{\sqrt{\frac{P}{r_i^{{\eta}}}}\sum_{k=1}^L a^{(i)}_{l, k} \tilde{s}_ {i,k}}}  + \bar{\textbf{w}}^T_{o,l}   \textbf{n}.
\end{align}

\noindent Conditioned on $r_i$ and $a_{l,k}^{(i)}$ $\forall \{i,k,l\}$, the lumped interference-plus-noise term in \eqref{eq:mimo_rxed_signal2} is Gaussian because of the Gaussianity of $\tilde{s}_{i,k}$, $\forall{i,k}$. This renders the well-known AWGN error probability expressions legitimate to conduct the averaging in step (i), in which the noise variance used in the AWGN-based expressions is replaced by the variance of the lumped interference and noise terms in \eqref{eq:mimo_rxed_signal2}. That is, the decoding error can be studied by using the AWGN expressions with the conditional SINR (i.e., conditioned on the channel gains and network geometry), followed by the averaging step (ii). The Gaussian signaling approximation leads to the following proposition.

\begin{proposition} \label{prop}
Consider a downlink MIMO cellular network with $N_t$ antennas at each BS and $N_r$ antennas at each UE in a Rayleigh fading environment with i.i.d. unit-mean channel power gains and Gaussian signaling approximation for the interfering symbols, then the per-data stream conditional SINR at the decoder of a generic UE after combining/equalization can be represented via the following equivalent SISO-SINR 
\begin{align} \label{instant}
\Upsilon =\frac{P r_o^{-\eta} \tilde{g}_o}{\sum_{r_i \in \tilde{\Psi}^o} P r_i^{-\eta} \tilde{g}_i +\mathcal{N}_o},
\end{align} 
\noindent where the random variables $\tilde{g}_o \sim \text{Gamma}\left(m_o,1 \right)$ and $\tilde{g}_i \sim \text{Gamma}\left(m_i,1 \right)$ capture the effect of MIMO precoding, combining, and equalization. The values of $m_o$ and $m_i$ are determined based on the number of antennas ($N_t$ and $N_r$), the number of multiplexed data streams per BS $(L)$ and the employed MIMO configuration as shown in Table~\ref{tab:summary}.

\end{proposition}
{
\begin{table}[h!]

\begin{center}
\begin{tabular}{|c| c| c| c| c | c|}
\hline 
 \textbf{MIMO Setup} & $L$ & $\boldsymbol{m_o}$ & $\boldsymbol{m_i}$ & Accuracy  & Proof \\  

\hline \hline 

 SIMO & $1$ & $N_r$ & $1$ & Exact & \emph{Lemma~\ref{lemma:SIMO} }  \\ 
\hline

  OSTBC & $N_s$ & $N_s N_r$ & $N_s$  & Exact & \emph{Lemma~\ref{lemma:OSTBC} } \\  
\hline 
  ZF-Rx & $N_t$ & $N_r-N_t+1$ & $N_t$ & Exact & \emph{Lemma~\ref{lemma:ZF_Rx} } \\  
\hline
 SDMA & $\mathcal{K}$ &  $N_t-\mathcal{K}+1$ & $\mathcal{K}$  & Approx. & \emph{Lemma~\ref{lemma:SDMA} }\\ 
\hline
MISO & $1$ &  $N_t$ & $1$  & Exact & \emph{Corollary~\ref{cor:MISO} }\\ 
\hline
SM-MIMO & $N_t$ & $N_r$ & $N_t$ & Approx. & \emph{Lemma~\ref{lemma:SM} }\\ 
\hline
\end{tabular}

\end{center}
\caption{SISO-equivalent gamma distribution parameters for various MIMO settings. }
\label{tab:summary}
\end{table}
}\normalsize
\begin{proof}
The detailed discussion and proof of each MIMO setup is given in Section~\ref{sec:mimo_schemes} in the corresponding lemma shown in Table~\ref{tab:summary}. Here, we just sketch a high-level proof of the proposition. The equivalent channel gains in \eqref{instant} are $\tilde{g}_o = \vert{g_o}\vert^2$ and $\tilde{g_i} = \vert{g_i}\vert^2=\sum_{k=1}^L\vert a^{(i)}_{l, k}\vert^2$, where ${g}_o$ is the random scale for the intended symbol due to precoding and combining/equalization as show in \eqref{eq:mimo_rxed_signal}. Since precoding and combining/equalization are usually in the form of linear combination of the channel power gains and that the channel gains have independent Gaussian distributions, both random variables $\tilde{g_o} $ and $\tilde{g_i}$ $\forall i$ are independent $\chi^2$-distributed with degrees of freedom equal to the number of linearly combined random variables, which depends on the number of antennas, precoding technique, and number of multiplexed data streams per BS. Note that the $\chi^2$ distribution for interfering channel gains $\tilde{g_i}$ is exact only if the precoding vectors in each BS are independent. In the case of dependent precoding vectors, the correlation is ignored and the $\chi^2$ distribution for $\tilde{g_i}$ is an approximation. Such approximation is commonly used in the literature for tractability~ \cite{eid_Mimo,  error_per2, STBC_harpret, Int_align_tony, Massive_harpret, MIMO_heath, CoMP_sakr, CoMP_martin, coordination1_martin, Mimo_ordering, Mimo_load_harpreet}, and is verified in Section~\ref{sec:sim_results}. 
Exploiting the one-to-one mapping between the  $\chi^2$ distribution and the gamma distribution, we follow the convention in  \cite{eid_Mimo,  error_per2, STBC_harpret, Int_align_tony, Massive_harpret, MIMO_heath, CoMP_sakr, CoMP_martin, coordination1_martin, Mimo_ordering, Mimo_load_harpreet} and use the gamma distribution, instead of the $\chi^2$ distribution, for $\tilde{g}_o$ and $\tilde{g}_i$. 
\end{proof}

 Proposition~\ref{prop} gives the equivalent SISO representation for the MIMO cellular network in which the effect of precoding, combining, and equalization of the employed MIMO scheme is abstracted by the random variables $\tilde{g_o}$ and $\tilde{g_i}$, $\forall i$ in \eqref{instant}. Hence, unified analysis and expressions for different KPIs and  MIMO configurations, respectively, are viable as shown in the next section.

\section{Unified Performance Analysis} \label{unified_ana}


Based on the Gaussian signaling approximation,  interference-plus-noise in \eqref{eq:mimo_rxed_signal2} is conditionally Gaussian. Hence, the decoding error performance of the MIMO scheme is studied by plugging the conditional SINR in \eqref{instant} with the appropriate channel gains (i.e., $\tilde{g}_o$ and $\tilde{g}_i$) in the corresponding AWGN-based decoding error expression, followed by an averaging over the channel gains and network geometry (i.e., step (ii)). Using the AWGN expression for the SEP for $M$-QAM modulation scheme  given in \cite{slim_comm_book}, the ASEP in MIMO cellular networks can be expressed as
\small
\begin{align} \label{assep}
\text{ASEP}\left(\Upsilon\right)= {w_{1} }  \mathbb{E}\left[\text{erfc}\left(\sqrt{\beta \Upsilon } \right)\right] + {w_{2} }  \mathbb{E}\left[\text{erfc}^2 \left(\sqrt{\beta \Upsilon} \right)\right],
\end{align}
\normalsize

\noindent   where $w_1= 2 \frac{\sqrt{M}-1}{\sqrt{M}}, w_2= -\left( \frac{\sqrt{M}-1}{\sqrt{M}}\right)^2$, and  $\beta= \frac{3}{2 \left( M-1\right)}$ are constellation-size specific constants\footnote{The error probability expression in \eqref{assep} can model other modulation schemes by just changing the modulation-specific parameters as shown in \cite{slim_comm_book}.}. The Gaussian signaling approximation is also the key that unifies the ASEP, outage probability, and ergodic rate analysis. This is because both outage and capacity are information theoretic KPIs that implicitly assume Gaussian codebooks, which directly lead to the conditional SINR in the form given by Proposition~\ref{prop}. Consequently, both the outage probability and ergodic capacity are also functions of the SINR in the form of \eqref{instant}, and are given by
\begin{align}\label{out_orig}
\mathcal{O} =  \mathbb{P}\left\{ \Upsilon <\theta \right\},
\end{align}
\vspace{-0.5cm}
and  
\begin{align}\label{rate_orig}
 \mathcal{R} =  \mathbb{E}\left[ \ln \left(1+ \Upsilon \right) \right].
\end{align}

According to step (ii) of the analysis given in Section~\ref{sec:gauss_apprx}, the expectations in \eqref{assep}, \eqref{out_orig}, and \eqref{rate_orig} are with respect to the network geometry and channel gains, which are evaluated via stochastic geometry analysis. Such expectations are usually expressed in terms of the Laplace transform (LT)\footnote{The LT of the interference is a short for the LT of the probability density function (PDF) of the interference, which is equivalent to the moment generating function but with negative argument.} of the aggregate interference power in \eqref{instant}, denoted as $\mathcal{I} =\sum_{r_i \in \tilde{\Psi}^o} P r_i^{-\eta} \tilde{g}_i $. The LT of the interference power in the SISO-equivalent SINR given in \eqref{instant} is characterized by the following lemma.
\begin{lemma} \label{the_LT}
Consider a cellular network with MIMO transmission scheme that can be represented via the equivalent SISO-SINR in \eqref{instant} and BSs modeled via a PPP with intensity $\lambda$, in which each BS transmits a symbol vector of length $L$ per channel use (pcu) with  symbols drawn from a zero-mean  unit-variance Gaussian distribution, then the LT of the interference power affecting an arbitrary symbol at a receiver located $r_o$ meters away from its serving BS is given by
\small
\begin{align}
 \!\!\!\!\!\!\!\!\! \mathcal{L}_{\mathcal{I} \vert r_o}\left(z\right)    = \exp{\left\{- \pi \lambda r_o^2 \left[ \left(_{2}F_{1}\left(\frac{-1}{b}, m_i; 1-\frac{1}{b}; -z P r_o^{-\eta} \right)-1 \right) \right] \right\}}, 
\end{align}
\normalsize

\noindent where ${_2}F_1(\cdot,\cdot;\cdot;\cdot)$ is the Gauss hypergeometric function \cite{abramowitz_stegun}. 

\begin{proof}
 Starting from the definition of the LT, we have
  \small
\begin{align}
 \!\!\!\!\!\!\!\!\! \mathcal{L}_{\mathcal{I} \vert r_o}\left(z\right)   &{=}  \mathbb{E} \left[\exp{\left\{-z \displaystyle \sum_{r_i \in \tilde{\Psi}^o} P r_i^{-\eta} \tilde{g}_i \right\} } \right] \overset{(b)}{=}\exp{\left\{ -2 \pi \lambda\int_{r_o}^{\infty}  \mathbb{E} \left[1- e^{-z P {x^{-\eta} \tilde{g}} }\right] x \mathrm{d}x\right\}} \notag \\
 &\overset{(c)}{=}\exp{\left\{ -2 \pi \lambda\int_{r_o}^{\infty}  \mathbb{E} \left[1- \frac{1}{\left( 1+z P x^{-\eta}\right)^{m_i}} \right] x \mathrm{d}x\right\}}, 
\end{align}
\normalsize
\noindent where $(b)$ follows from the PGFL of the PPP \cite{martin_book}, and $(c)$ follows from the LT of  the gamma distributed channel gains with shape parameter $m_i$ and unity scale parameter.
\end{proof}
\end{lemma} 
Using Proposition~\ref{prop} and Lemma~\ref{the_LT}, we arrive to the unified MIMO expressions for the ASEP, outage probability, and ergodic rate in the following theorem.
\normalsize
\begin{figure*}[t]
\footnotesize{
\begin{align}
\!\!\! \!\!\! \text{ASEP}& \approx w_1 \left[ 1- \frac{\Gamma\left(m_o+\frac{1}{2}\right)}{\Gamma \left(m_o \right)} \frac{2}{\pi} \int_{0}^{\infty} \int_{0}^{\infty} 2 \pi \lambda_B x e^{- \pi \lambda_B x^{2}}\frac{1}{\sqrt{z}} e^{-z \left(1+  \frac{ m_o \mathcal{N}_o   x^{\eta}}{ \beta P} \right)} \:  _{1}F_{1} \left(1-m_o; \frac{3}{2}; z \right) \cdot \mathcal{L}_{\mathcal{I} \vert x}\left( \frac{m_o z x^{\eta}}{ \beta P } \right) \mathrm{d}x \mathrm{d}z\right]  \notag \\ 
&+ w_2 \left[ 1- \frac{4 m_o}{\pi} \int_{0}^{\infty} \int_{0}^{\infty} 2 \pi \lambda_B x e^{- \pi \lambda_B x^{2}} e^{-z \left(  \frac{ m_o \mathcal{N}_o  x^{\eta}}{ \beta P}\right)} \:  \int_0^{\frac{\pi}{4}} {_{1}F_{1}}\left(m_o+1; 2 ;\frac{-z}{\text{sin}^2 \vartheta} \right)  \frac{1}{\text{sin}^2 \vartheta}\cdot \mathcal{L}_{\mathcal{I} \vert x}\left( \frac{m_o z x^{\eta}}{ \beta P } \right)  \mathrm{d}\vartheta \mathrm{d}x \mathrm{d}z\right] . 
\label{eq:asep_general}
\end{align}
} \hrulefill
\end{figure*}

\begin{theorem} \emph{Unified Analysis:} Consider a cellular network with MIMO transmission scheme that can be represented via the equivalent SISO-SINR in \eqref{instant} and BSs modeled via a PPP with intensity $\lambda$, in which each BS transmits a symbol vector of length $L$ pcu with symbols drawn from an equiprobable unit-power square quadrature amplitude modulation ($M$-QAM) scheme, then the ASEP for an arbitrary symbol is approximated by \eqref{eq:asep_general}, where ${_1}F_1(\cdot;\cdot;\cdot)$ is the Kummer confluent hypergeometric function \cite{abramowitz_stegun}. 

For an interference-limited scenario, the probability that the SIR for an arbitrary symbol goes below a threshold $\theta$ is given by
\small{
\begin{align} \label{out}
\mathcal{O}\left( \theta\right) &=   1- \!\int_{0}^{\infty}\!\!\!\!\!\! 2 \pi \lambda_B \: x e^{-\pi \lambda_B x^2}   \sum_{j=0}^{m_o-1}  \frac{\left(-1\right)^j}{j !} \frac{\mathrm{d}^j}{\mathrm{d}z^j} \left(\frac{\theta  r_o^{\eta}}{ P } \right)^{j}  \mathcal{L}_{\mathcal{I} \vert x}\left(z\right) \Bigg \vert_{z= \frac{ \theta  x^{\eta}}{ P}} \!\!\!\!\!\!\! \!\!\!\!\mathrm{d}x,
\end{align}
} \normalsize
\noindent and the ergodic rate for an arbitrary data stream is given by
\small{
\begin{align} \label{rate}
\!\!\!\!\!\mathcal{R}= \int_0^{\infty} \!\!\!\! \int_0^{\infty}\!\!\! 2 \pi \lambda_B \:x e^{- \pi \lambda_B x^2} \mathcal{L}_{\mathcal{I} \vert x}\left( z\right)\left(\frac{1-\left( 1+  z\right)^{-m_o}}{z}\right) \mathrm{d}z \mathrm{d}x,
\end{align}
} \normalsize

\noindent where $ \mathcal{L}_{\mathcal{I} \vert x}\left(z\right)$ in \eqref{eq:asep_general}, \eqref{out}, and \eqref{rate} is the LT given in Lemma~\ref{the_LT} when replacing $r_o$ with $x$. 

\begin{proof}
See Appendix \ref{app:outage}.
\end{proof}

\label{theo:unified}
\end{theorem}

The ASEP given in \eqref{eq:asep_general} is an approximation due to the Gaussian signaling approximation used in \eqref{eq:mimo_rxed_signal2}. However, the outage probability and ergodic rate are both exact because both are typically derived based on the Gaussian codebooks. The outage probability in \eqref{out} is given for interference-limited networks for tractability, which is a common assumption in cellular network because the interference term usually dominates the noise. Both equations \eqref{out} and \eqref{rate} are approximations  in cases of SDMA and SM-MIMO due to the approximate estimation of the interference as shown in Table~\ref{tab:summary}. 

The ASEP expression given in \eqref{eq:asep_general} presents three advantages over the ASEP expressions given in \cite{eid_Mimo}. First, \eqref{eq:asep_general} provides a unified ASEP expression for all considered MIMO schemes. Second, the ASEP is characterized based on the LT given in Lemma~\ref{the_LT}, which is the same LT used for characterizing the outage probability and ergodic rate. Third, the computational complexity to evaluate \eqref{eq:asep_general} is less than the complexity of the ASEP expressions in \cite{eid_Mimo}. The reduced complexity of \eqref{eq:asep_general} is because it includes a single  hypergeometric function in the exponential term while the expressions for the ASEP in \cite{eid_Mimo} include summations of hypergeometric functions inside the exponential term. In Section~\ref{sec:sim_results}, we show that the advantages presented by \eqref{eq:asep_general} does not sacrifice the model accuracy when compared to the ASEP in \cite{eid_Mimo}, specially that the ASEP in \cite{eid_Mimo} is also approximate in case of  SDMA and SM-MIMO schemes.


\subsection{The Effect of Temporal Correlation on Retransmissions}

Lemma~\ref{the_LT} gives the LT of the inter-cell interference measured at a generic location at an arbitrary point of time in a MIMO cellular network and Theorem~\ref{theo:unified} gives the performance of a MIMO cellular link experiencing such inter-cell interference. The network performance with retransmission cannot be directly deduced from Lemma~\ref{the_LT} and  Theorem~\ref{theo:unified} due to the temporal interference correlation.  Despite that we assume that the channel gains, due to fading, independently change from one time slot to another, the interference across time at a given location is correlated due to the fixed locations of the set of interfering BSs. To incorporate the effect of retransmission into the analysis, the temporal correlation of the interference should be characterized via the joint LT of the interferences across different time slots, as given by the following lemma.

\begin{lemma} \label{joint_LT}
Consider a cellular network with MIMO transmission scheme that can be represented via the equivalent SISO-SINR in \eqref{instant} and BSs modeled via a PPP with intensity $\lambda$ and activity factor {$p$}, the joint LT of the interferences at a given location at two different time slots, denoted by $\mathcal{I}_1$ and $\mathcal{I}_2$, such that the interfering BSs may use different MIMO scheme across time, is given by
\begin{align}
\mathcal{L}&_{\mathcal{I}_1, \mathcal{I}_2  \vert r_o}\left(z_1,z_2 \right) 
= \exp\left\{ - \pi   \lambda r_o^2 \left[  p \left( \mathcal{F}_1 \left( \frac{-2}{\eta}; m_{i,1}, m_{i,2}; 1-\frac{2}{\eta}; -z_1 r_o^{-\eta}, -z_2 r_o^{-\eta}\right)+1 \right)\right.  \right. \notag \\
&\left.   \left. + (1-p) \:\: \left( _{2}F_{1}\left(\frac{-1}{b}, m_{i,2}; 1-\frac{1}{b}; -z P r_o^{-\eta} \right)  +  {}_{2}F_{1}\left(\frac{-1}{b}, m_{i,1}; 1-\frac{1}{b}; -z P r_o^{-\eta} \right) \right)-2 \right] \right\},
\end{align}

\noindent where $m_{i,1}$ and $m_{i,2}$ are the rates of the Gamma distributed equivalent channel gains (given in Table~\ref{tab:summary}) corresponding  to the employed MIMO scheme at the first and second time slots, respectively, and $\mathcal{F}_1\left( \cdot; \cdot, \cdot; \cdot; x, y \right)$ is the Appell Hypergeometric function, which extends the hypergeometric function to two variables $x$ and $y$ \cite{nist}. 
\end{lemma}
\begin{proof}
See Appendix \ref{app:joint_LT1}.
\end{proof}

The average coverage probability (defined as $1- \mathcal{O}\left(\theta \right)$) with retransmissions and independent signal decoding is given by 
\begin{align} 
\mathcal{P}_c \left(\theta \right)= \mathbb{P}\left( \bar{\Upsilon}_{1} > \theta\right)+ \mathbb{P}\left( \bar{\Upsilon}_{2} > \theta\right)-\mathbb{P}\left( \bar{\Upsilon}_{1} > \theta, \bar{\Upsilon}_{2} > \theta\right),
\end{align}

\noindent where $\bar{\Upsilon}_{1} $ and $\bar{\Upsilon}_{2}$ are the SIRs at the first and second transmissions.  
 Using the joint LT in Lemma~\ref{joint_LT}, the average coverage probability with retransmission in MIMO cellular network is given by the following theorem.

\begin{theorem} 
\label{theo:joint_LT}
Consider a cellular network with MIMO transmission scheme that can be represented via the equivalent SISO-SINR in \eqref{instant} and BSs modeled via a PPP with intensity $\lambda$, the SIR coverage probability for a generic UE with retransmission such that the serving BS and interfering BSs may use different MIMO schemes across time, is given by  \eqref{eq:Pc},
\begin{figure*}
\small{
\begin{align}
\label{eq:Pc}
\mathcal{P}_c  \left(\theta \right)
&= \int_{0}^{\infty}\!\!\!\!\!\! 2 \pi \lambda_B \: x e^{-\pi \lambda_B x^2} \left[  \sum_{j_1=0}^{m_{o,1}-1} \frac{\left(-1\right)^{j_1}}{j_1 !  } \bar{  \theta }^{j_1 } \frac{\partial^{j_1 }}{\partial z_1^{j_1}} \mathcal{L}_{\mathcal{I}_1 \vert x}\left(z_1 \right) \bigg \vert_{z_1=\bar{\theta}} + \sum_{j_2=0}^{m_{o,2}-1} \frac{\left(-1\right)^{j_2}}{j_2 !  } \bar{  \theta }^{j_2 } \frac{\partial^{j_2 }}{\partial z_2^{j_2}} \mathcal{L}_{\mathcal{I}_2 \vert x}\left(z_2 \right) \bigg \vert_{z_2=\bar{\theta}}\right. \notag \\
 & \quad \quad \quad \quad\quad \quad \quad \quad \quad \left. -  \sum_{j_1=0}^{m_{o,1}-1}  \sum_{j_2=0}^{m_{o,2}-1} \frac{\left(-1\right)^{j_1+j_2}}{j_1 ! j_2 !}   \bar{  \theta }^{j_1+j_2} \frac{\partial^{(j_1+j_2)}}{\partial z_1^{j_1} \partial z_2^{j_2}} \mathcal{L}_{\mathcal{I}_1, \mathcal{I}_2 \vert x}\left(z_1,z_2 \right) \bigg \vert_{z_1= z_2=\bar{\theta}} \right] \mathrm{d}x.
\end{align}
} \hrule
\end{figure*}
\normalsize
where $  \bar{  \theta }=\frac{\theta x^\eta}{P}$. 
\begin{proof}
See Appendix \ref{app:joint_LT2}.
\end{proof}
\end{theorem}

Before giving numerical results and insights obtained from the developed mathematical model, we first illustrate how the equivalent SISO-SINR model given in Proposition~\ref{prop} holds for the considered MIMO schemes.

\section{Characterizing MIMO Configurations}
\label{sec:mimo_schemes}

This section details the methodology to abstract different MIMO configuration via the equivalent SISO model given in  Proposition~\ref{prop} with parameters given in Table~\ref{tab:summary}. In order to conduct the analysis for the different MIMO setups, we first need to define the set $\left\{\breve{\textbf{H}}\right\}$ as the set of channel matrices that affect the aggregate interference signals due to precoding and/or combining. For instance, due to precoding, combining, and equalization, the interference from the $i^{th}$ interfering BS is multiplied by $\textbf{W}_o \textbf{H}_i \textbf{V}_i$, and hence, $\left\{\breve{\textbf{H}}\right\}=\{\textbf{H}_o, \tilde{\textbf{H}}_i\}$, where ${\textbf{H}}_o$, and $\tilde{\textbf{H}}_i$ are the channel matrices between, respectively, the intended BS and the test user, and the $i^{th}$ interfering BSs and its associated users. 
The methodology to characterize the distribution of the equivalent channel gains are given in the following steps:
\begin{enumerate}

\item  \textbf{SNR characterization:} $\tilde{g}_o$ is first characterized by projecting the signal of the intended data-stream on the null-space of the signals of the other data streams that are multiplexed by the intended BSs. Note that we may manipulate the resultant SNR such that the noise variance is not affected by any random variable as in \eqref{instant} and the projection effect is contained in $\tilde{g}_o$ and $\tilde{g}_i$ only. 

\item  \textbf{Per-stream equivalent channel gain representation:} $\tilde{g}_i$ from each interfering BS is characterized based on the manipulation done in the SNR characterization in the previous step and characterizing $\vert a_{l,k}^{(i)} \vert^2$ given in  \eqref{eq:mimo_rxed_signal2}. Note that $\vert a_{l,k}^{(i)} \vert^2$ is characterized based on $\left\{\breve{\textbf{{H}}}\right\}$  which captures the channel gain matrices involved in precoding the signal at the ${i^{th}}$ BS and combining the interfering symbols at the test UE.

\end{enumerate}

Based on the aforementioned two steps, the equivalent SISO-SINR given in Proposition~\ref{prop} for the MIMO schemes given in Table~\ref{tab:summary} is illustrated in this section.

\subsubsection{\textbf{Single-Input-Multiple-Output (SIMO) systems}}
for a SIMO system, receive diversity is achieved using one transmit antenna (i.e., $L=N_t=1$) and $N_r$ receive antennas. Since $N_t=1$, then the intended and interfering channel vectors are denoted by $\textbf{h}_o$ and $\textbf{h}_i \in \mathbb{C}^{N_r \times 1}$, respectively. By employing Maximum Ratio-Combining (MRC) strategy to combine the received signals, then $\bar{\textbf{w}}_o^T = \textbf{h}_o^{H}$. The equivalent SISO channel gains are given by the following lemma.

\begin{lemma}
\label{lemma:SIMO}
For a receive diversity SIMO setup technique, the Gamma distribution parameters for the equivalent intended and interfering channel gains are given by $\boldsymbol{m_o=N_r}$ and $\boldsymbol{m_i=1}$, respectively.

\begin{proof}
See Appendix \ref{app:all_lemmas}.
\end{proof}
\end{lemma}

\subsubsection{\textbf{Orthogonal Space-Time Block Coding (OSTBC)}} Let the orthogonal space-time block codes be transmitted over $T$ time instants, and only $N_s \leq N_t$ transmit antennas are active per time instant. The received signal is equalized via the equalizer $\textbf{W}_o= \frac{\textbf{H}_{\text{eff}}^{H}} {\| \textbf{H}_o\|_{\text{F}}}$, where $\textbf{H}_{\text{eff}}$ is the effective intended channel matrix depending on the employed orthogonal code~\cite{STBC_harpret, eid_Mimo}. Since no precoding is applied then $\textbf{V}_o=\textbf{V}_i= \textbf{I}_{N_t}$. The equivalent SISO channel gains are given by the following lemma.

\begin{lemma}
\label{lemma:OSTBC}
A space-time encoder is employed at the network BSs. Then, the Gamma distribution parameters are given as $\boldsymbol{m_o=N_s N_r}$ and $\boldsymbol{m_i=N_s}$.

\begin{proof}
See Appendix~\ref{app:all_lemmas}.
\end{proof}

\end{lemma}

\subsubsection{\textbf{Zero-Forcing beamforming with ML Receiver (ZF-Rx)}} ZF is a low-complexity suboptimal, yet efficient, technique to suppress interference from other transmitted symbols in the network. In order to recover the distinct transmitted streams, the received signal is multiplied by the equalizing matrix $\textbf{W}_o=\left(\textbf{H}_o^H \textbf{H}_o\right)^{-1} \textbf{H}_o^H$ representing the pseudo-inverse of the intended channel matrix $\textbf{H}_o$, whereas we assume no precoding at the transmitters side, i.e., $\textbf{V}_o=\textbf{V}_i=\textbf{I}_{N_t}$. The equivalent SISO channel gains are given by the following lemma.

\begin{lemma}
\label{lemma:ZF_Rx}
By employing a ZF-Rx such that $L=N_t$ distinct streams are being transmitted from the BSs, it can be shown that $\boldsymbol{m_o=N_r-N_t+1}$ and $\boldsymbol{m_i=N_t}$.

\begin{proof}
See Appendix~\ref{app:all_lemmas}.
\end{proof}
\end{lemma}




\subsubsection{\textbf{Space-Division Multiple Access (SDMA)}}
SDMA is used to accommodate more users on the same resources to increase the network capacity. In this case, we consider that each BS is equipped by $N_t$ transmit antennas and applies ZF transmission to simultaneously serve $\mathcal{K}$ single-antenna UEs that are independently and randomly distributed within its coverage area. To avoid rank-deficiency, we let $N_t \geq \mathcal{K}$, and hence, the number of data streams $L=\mathcal{K}$. A ZF-precoding in the form of $\textbf{V}_o=\left[ \textbf{v}_1 , \textbf{v}_2 \cdots \textbf{v}_{\mathcal{K}}\right]$ such that
 $\textbf{v}_{ l}=\frac{ \textbf{q}_{ l}}{\| \textbf{q}_{ l}\|}$ and $ \textbf{q}_{ l}$ is defined as the $ l^{th}$ column of ${\textbf{Q}}=\textbf{H}_o^H \left( \textbf{H}_o \textbf{H}_o^H\right)^{-1}$ is applied by the test BS and no combining is applied  at the single antenna test UE, and heck, $\textbf{W}_o=\textbf{I}_{\mathcal{K}}$. The interfering BSs apply the same precoding and combining strategy, and hence, the interfering precoding matrices are in the form  $\textbf{V}_i=\left[ \textbf{v}_{i,1} , \textbf{v}_{i,2} \cdots \textbf{v}_{i,\mathcal{K}}\right]$ such that
 $\textbf{v}_{i, k}=\frac{ \textbf{q}_{ k}}{\| \textbf{q}_{ k}\|}$ and $ \textbf{q}_{k}$  is the $k^{th}$ column of $\textbf{Q}_i=\tilde{\textbf{H}}_i^H \left(\tilde{ \textbf{H}}_i \tilde{\textbf{H}}_i^H\right)^{-1}$, note that $\tilde{\textbf{H}}_i\neq \textbf{H}_i$ is the interfering channel matrix towards the corresponding intended users. The equivalent SISO channel gains are given by the following lemma.

\begin{lemma}
\label{lemma:SDMA}
In a multi-user MIMO setup, the corresponding Gamma distributions parameters are given by $\boldsymbol{m_o= N_t - \mathcal{K}+1}$, and $\boldsymbol{m_i \approx \mathcal{K}}$. 
\begin{proof}
See Appendix~\ref{app:all_lemmas}.
\end{proof}
 \end{lemma}

\begin{corollary}{\textbf{Single-User Beamforming (SU-BF):}}
\label{cor:MISO}

The SDMA scenario reduces to SU-MISO setting if the number of served users in the network is $\mathcal{K}=1$. Hence, $m_o=N_t$ and $m_i=1$.
\end{corollary}


\subsubsection{\textbf{Spatially Multiplexed MIMO (SM-MIMO) systems}} for the sake of completeness, we also consider a spatially multiplexed MIMO setup with optimum joint maximum likelihood receiver. This case is important because it represents the benchmark for ZF decoding.  Note that the analysis in this case is slightly different from the aforementioned schemes since joint detection is employed. Nevertheless, it can be represented via the equivalent SISO-SINR given in Proposition~\ref{prop}. Due to joint detection, no precoding/combining is applied such that $\textbf{W}_o=\textbf{V}_o=\textbf{V}_i=\textbf{I}_{N_t}$. To analyze this case, we define the error vector $\textbf{e}\left( \textbf{s},\hat{\textbf{s}}\right)=\textbf{s}-\hat{\textbf{s}}$ as the distance between $\textbf{s}$ and $\hat{\textbf{s}}$ and hence we derive the APEP, which is then used to approximate the ASEP as shown in the following lemma.
\begin{lemma}{} 
\label{lemma:SM}
For a SM-MIMO transmission, the Gamma distribution parameter for the equivalent intended channel gains is given by $\boldsymbol{m_o=N_r}$, while for the equivalent interfering links is given by $\boldsymbol{m_i=N_t}$. Furthermore, the averaged PEP over the distance distribution of $r_o$ is
\small{
\begin{align}
 \text{APEP}(\|\textbf{e}\|)\approx 1- \frac{\Gamma\left(m_o+\frac{1}{2}\right)}{\Gamma \left(m_o \right)} \frac{2}{\pi} \int_{0}^{\infty} \int_{0}^{\infty} & 2 \pi \lambda_B x e^{- \pi \lambda_B x^{2}}\frac{1}{\sqrt{z}} e^{-z \left(1+  \frac{ m_o \mathcal{N}_o   x^{\eta}}{4  \|\textbf{e}\|^2 P} \right)} \:  _{1}F_{1} \left(1-m_o; \frac{3}{2}; z \right) \times \notag \\
&\mathcal{L}_{\mathcal{I} \vert x}\left( \frac{m_o z}{4 \|\textbf{e}\|^2  P x^{-\eta}} \right) \mathrm{d}x \mathrm{d}z.
\end{align} 
}\normalsize
Consequently, using the Nearest Neighbor approximation \cite{goldsmith}, where there are $M$ equiprobable symbols, then
\begin{align}
 \text{ASEP}\approx N_{\|\textbf{e}\|_\text{min}} \text{ APEP}\left(\|\textbf{e}\|_\text{min}\right),
\end{align}
where $N_{\|\textbf{e}\|_\text{min}}$ is the number of constellation points having the minimum Euclidean distance denoted by $\displaystyle\min_{\textbf{s},\hat{\textbf{s}}} \|\textbf{e}\left(\textbf{s},\hat{\textbf{s}}\right)\|$ among all possible pairs of transmitted symbols, and hence is a modulation-specific parameter.

\begin{proof}
See Appendix \ref{app:all_lemmas}.
\end{proof}
\end{lemma}


\section{Numerical and Simulation Results}
\label{sec:sim_results}

In this section, we verify the validity and accuracy of the proposed unified model and discuss the potential of such unified framework for designing cellular networks. 
Unless otherwise stated, the simulations setup is as follows. The BSs transmit powers ($P$) vary while $\mathcal{N}_o$ is kept constant to vary the transmit SNR, the path-loss exponent $\eta=4$, the noise power $\mathcal{N}_o=-90$ dBm, the BSs intensity $\lambda_B=10\; \text{BSs}/\text{km}^2$, $\lambda_u=20\; \text{users}/\text{km}^2$. The desired symbols are modulated using square quadrature amplitude modulation (QAM) scheme, with a constellation size $M$.
 
\subsection{Proposed model validation}

We validate the derived ASEP for the proposed model via Monte-Carlo simulations, in Fig. \ref{fig:all_validation_ASEP}, for the network setup detailed as: activity factor $p=1$ and $(a)$ SIMO:  $N_t=1$ and $N_r=3$, $(b)$ the OSTBC: we consider a $2 \times 2$ Alamouti code, $(c)$ zero-forcing with ML receiver with $N_t=2, N_r=5$, $(d)$ the SDMA multi-user setting using a zero-forcing precoder at the transmitters side with $\mathcal{K}=3$ single-antenna users to be served by BSs equipped with $N_t=5$ antennas. The figure further verifies the accuracy of the Gaussian signaling approximation and the developed ASEP model, in which the analytic ASEP expressions perfectly match the simulations. {Figs. \ref{fig:outageReTx_comp} and \ref{fig:increm_mo} validate Theorem 1 and Theorem 2 for the outage before and after retransmission against Monte-Carlo simulations.  Fig. \ref{fig:outageReTx_comp} shows the time diversity loss due to interference temporal correlation when compared to the independent interference scenario.} The figure shows that assuming independent interference across time is too optimistic, while considering the interference correlation degrades the retransmission performance. Thus, it is possible for the network operators to exploit more diversity in the second transmission to enhance the retransmission performance. Fig. \ref{fig:increm_mo}, shows the effect of incremental diversity in the second transmission on the outage performance for $m_o=2$ and {$m_i=2$}. The figure shows that adjusting the MIMO configuration such that $m_o =5$ in the second transmission compensates for the temporal correlation effect and achieves the same performance as independent transmission (e.g., up to $3$ dB SIR improvement can be achieved).

\normalsize
\begin{figure*}[t!]
    \centering
   \begin{subfigure}[t]{0.31\textwidth}
  \centerline{\includegraphics[width=  1.9in]{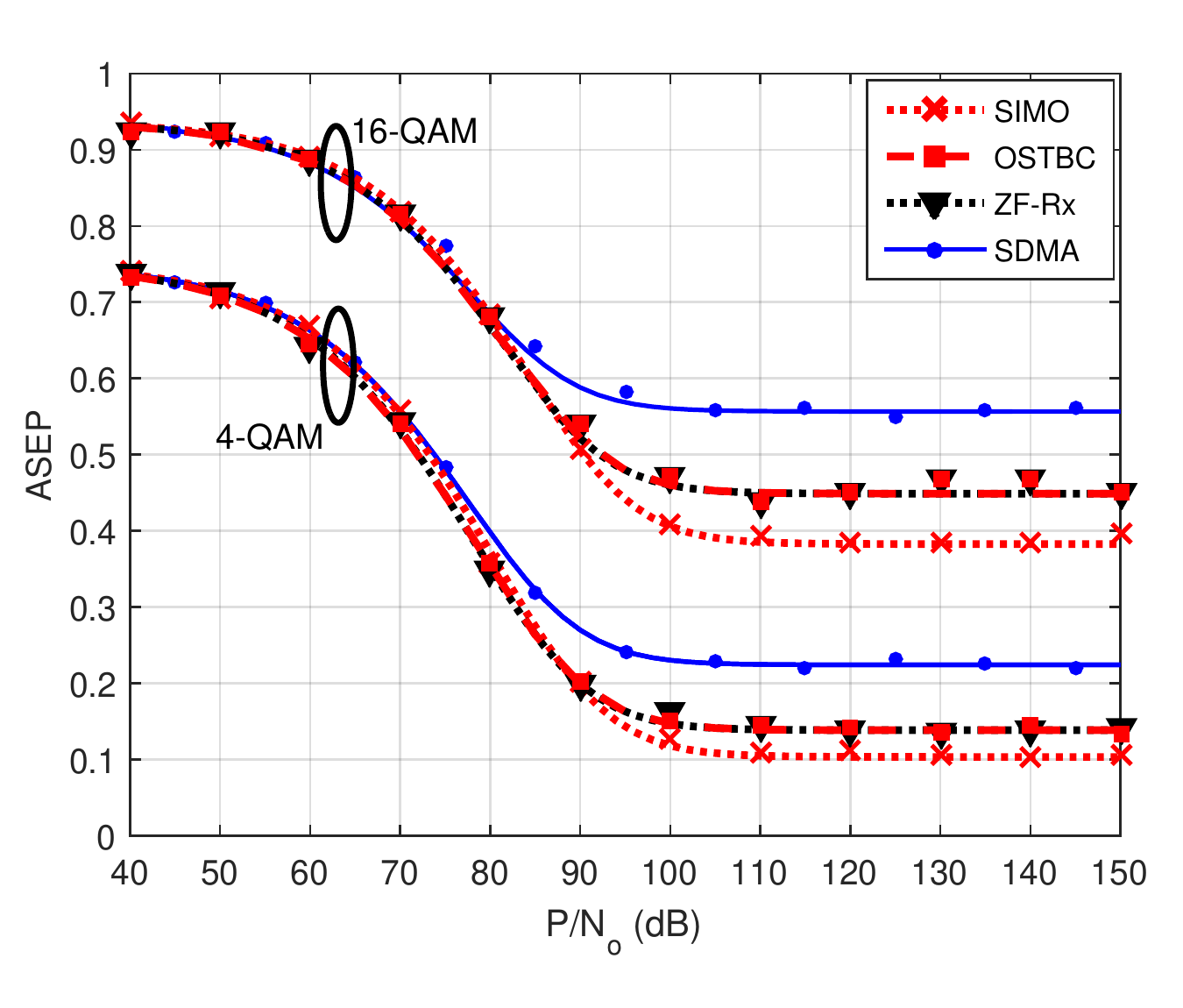}}
      \caption{\,ASEP against $\frac{P}{\mathcal{N}_o}$.}
\label{fig:all_validation_ASEP}
    \end{subfigure}
~
    \begin{subfigure}[t]{0.31\textwidth}
  \centerline{\includegraphics[width=  1.9in]{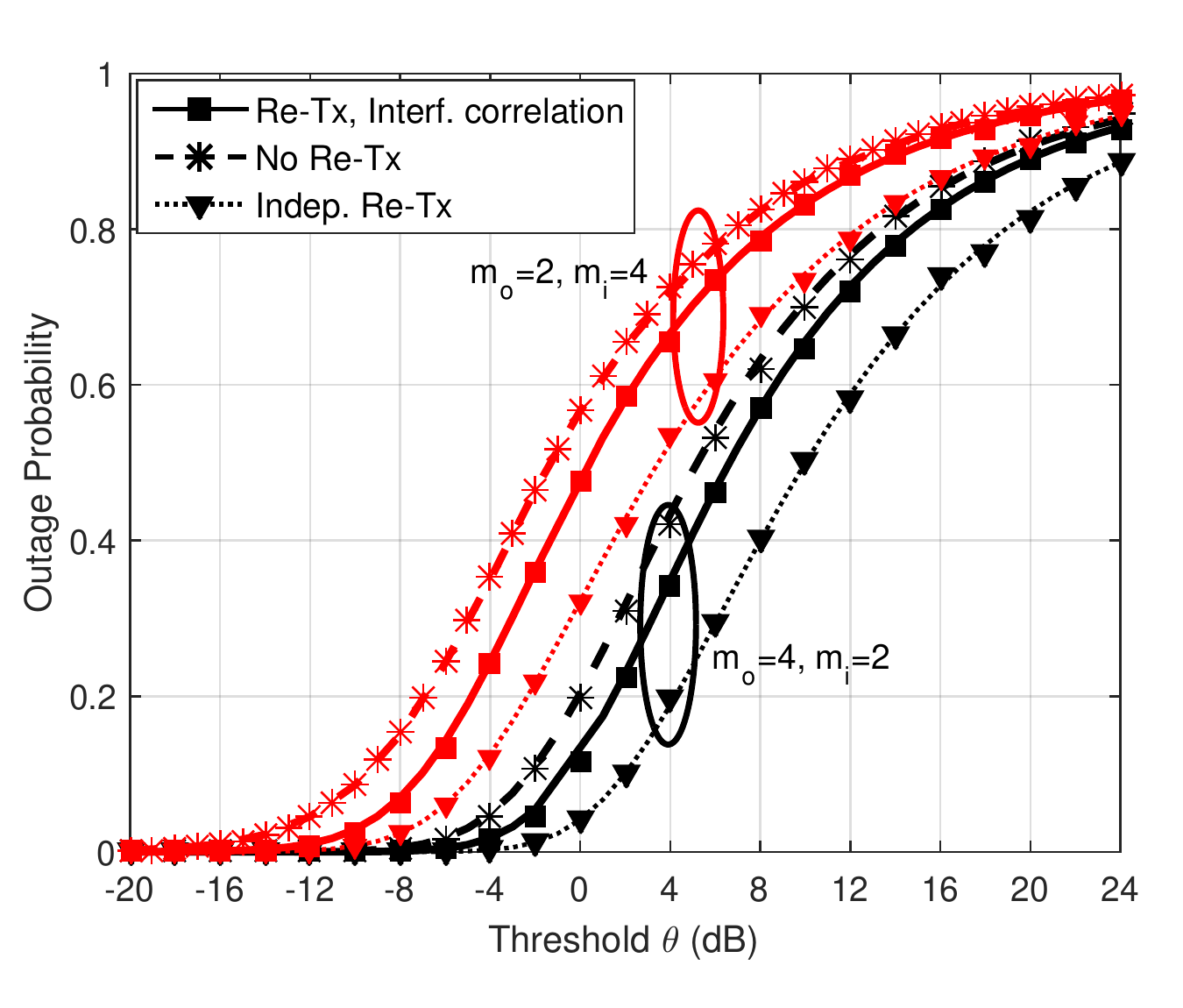}}
      \caption{\,The effect of interference correlation for different $m_o$ and $m_i$.}
\label{fig:outageReTx_comp}
    \end{subfigure}
    ~ 
    \begin{subfigure}[t]{0.31\textwidth}
       \centerline{\includegraphics[width=  1.9in]{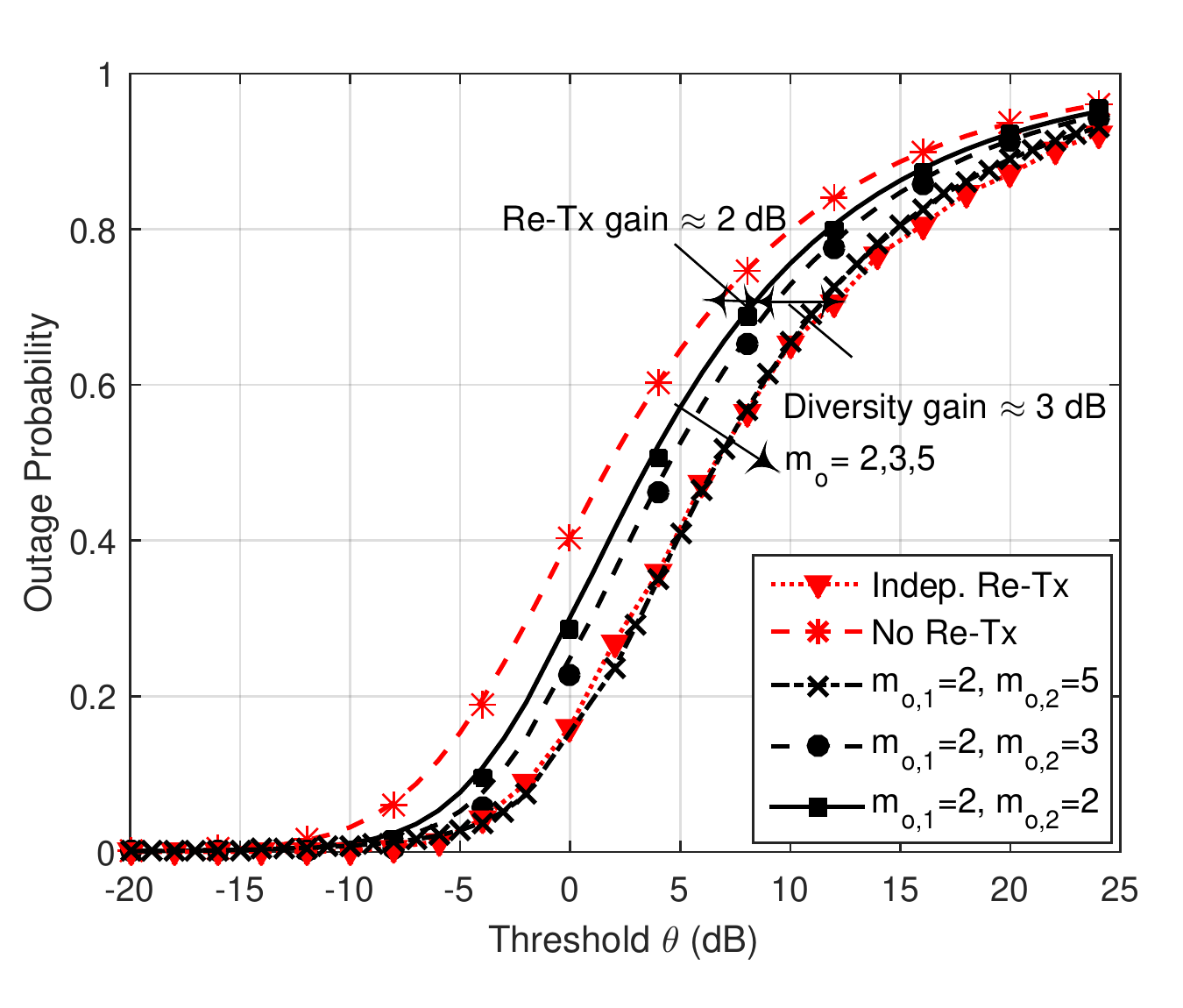}}\caption{\,Incremental diversity for the same inter-cell interference.}
\label{fig:increm_mo}
    \end{subfigure}
    \caption{ASEP and Outage probability performance validation. Lines represent the proposed analysis and markers represent Monte-Carlo simulations.}
    \label{retrans}
\end{figure*}




\begin{figure*}[t!]
    \centering
    \begin{subfigure}[t]{0.45\textwidth}
  \centerline{\includegraphics[width=  2.6 in]{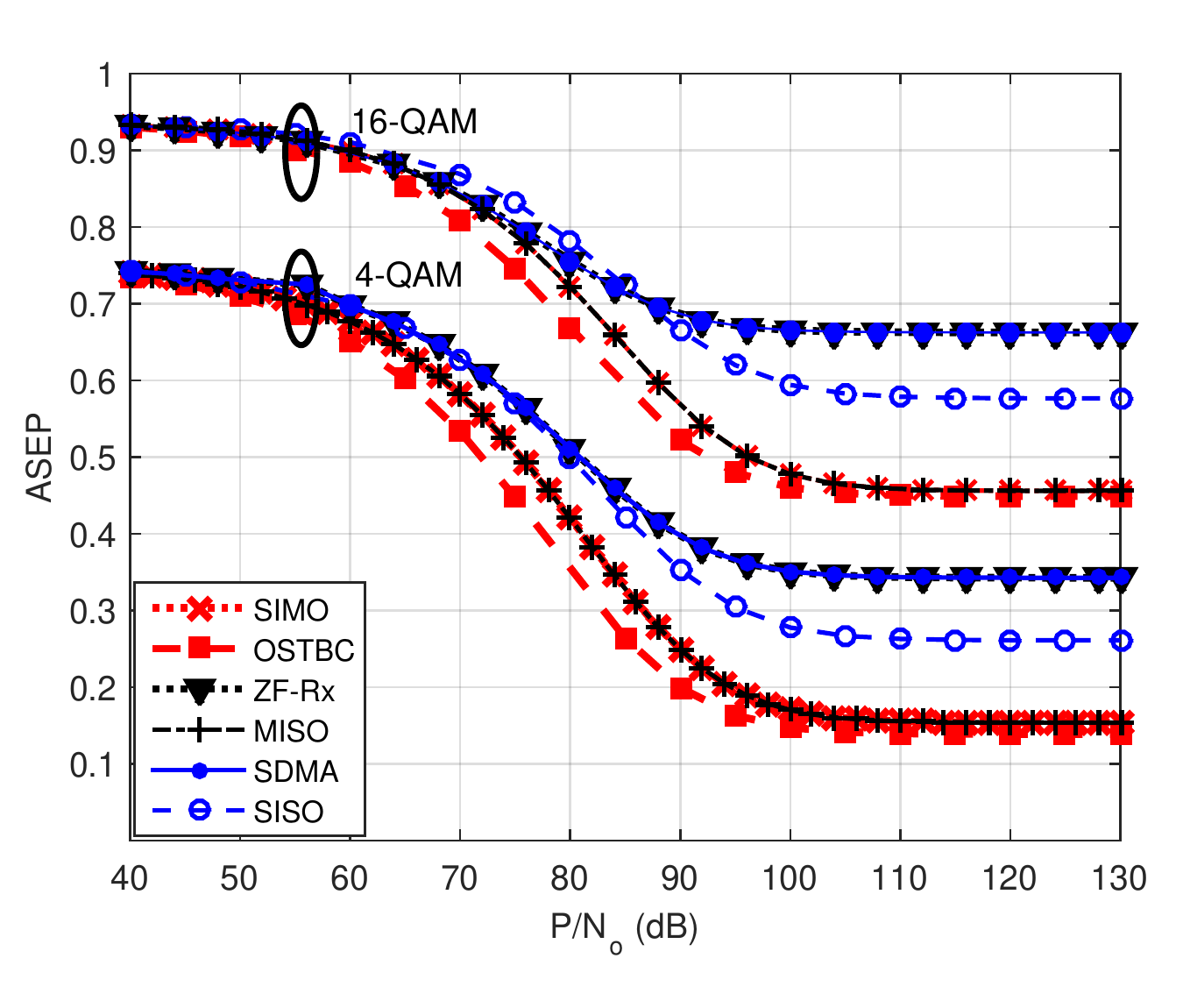}}
      \caption{\,ASEP for 4-QAM and 16-QAM.}
\label{fig:asep_Nr2_Nt2}
    \end{subfigure}
    ~ 
    \begin{subfigure}[t]{0.45\textwidth}
       \centerline{\includegraphics[width=  2.6 in]{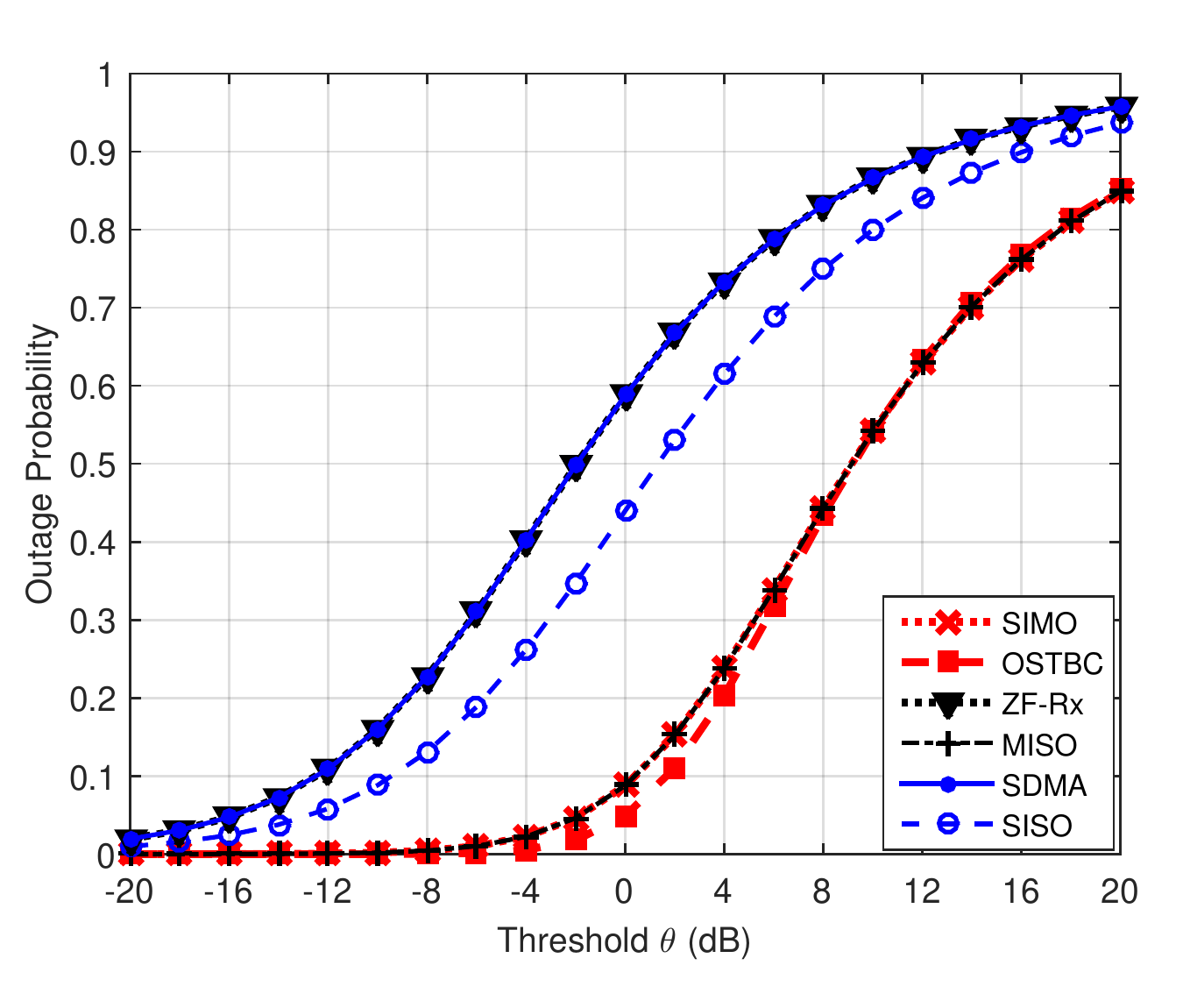}}\caption{\,Outage Probability versus SIR threshold $\theta$.}
\label{fig:out_Nr2_Nt2}
    \end{subfigure}

    \caption{ASEP and Outage probability performance for the different MIMO setups using the same number of antennas $N_t=2$ and $N_r=2$, at $p=1$.}
\label{fig:comparison_Nt2_Nr2}
\end{figure*}

\normalsize
\begin{figure*}[t!]
    \centering

    \begin{subfigure}[t]{0.31\textwidth}
       \centerline{\includegraphics[width=  1.9in]{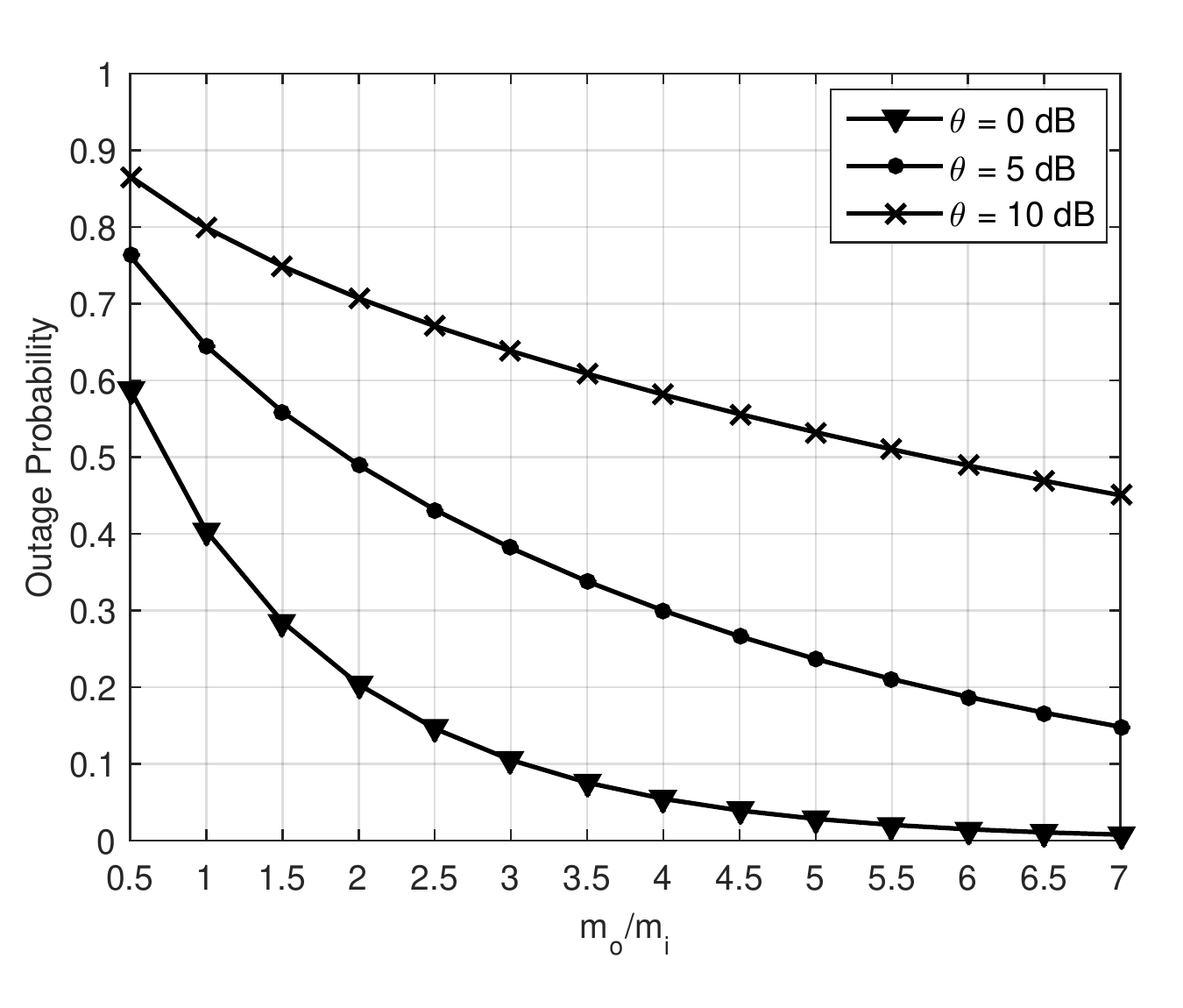}}\caption{\,Outage Probability for $\theta=0, 5, 10$ dB.}
\label{fig:out_unified}
    \end{subfigure}
~ 
 \begin{subfigure}[t]{0.31 \textwidth}
  \centerline{\includegraphics[width=  1.9in]{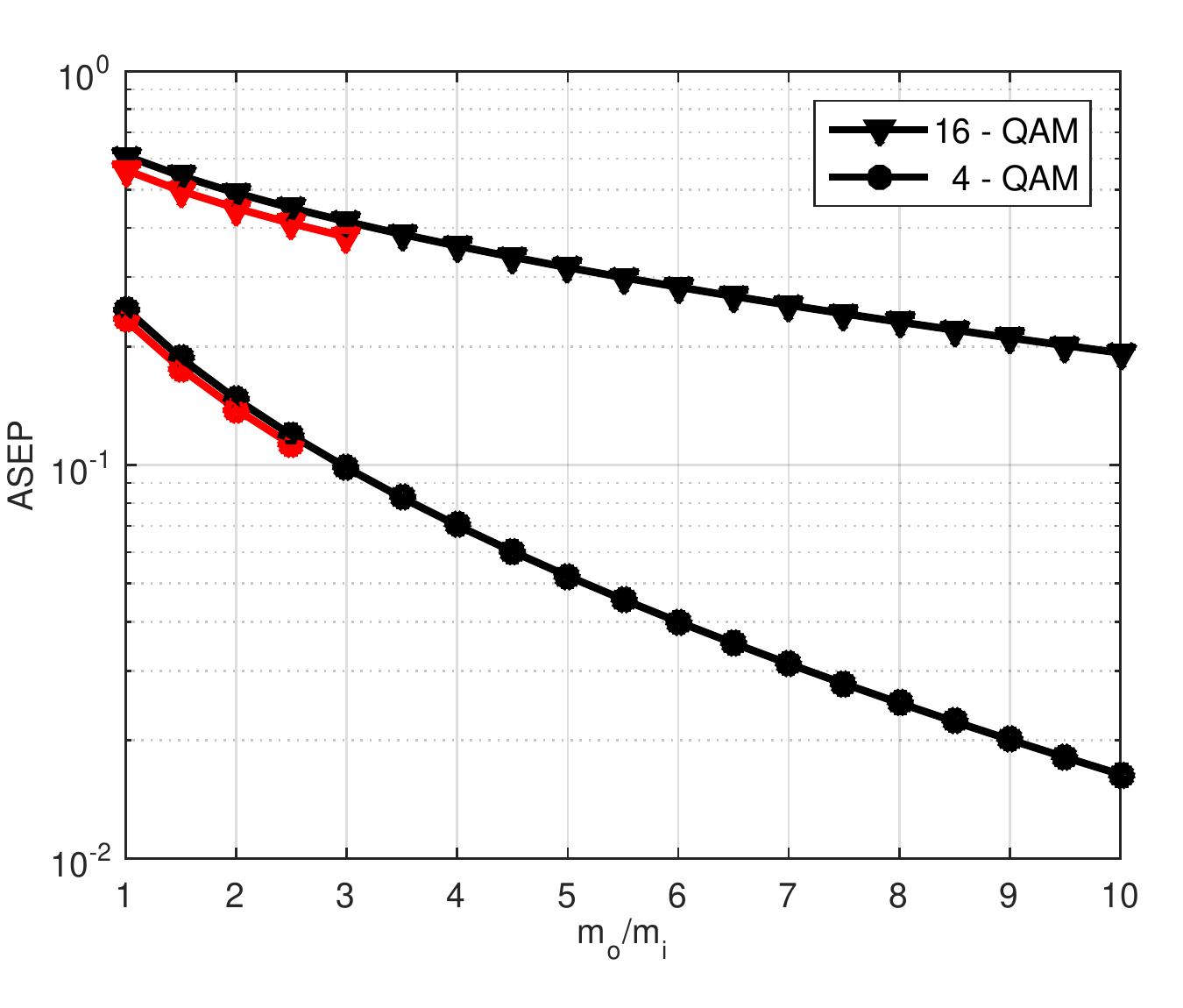}}
      \caption{\,ASEP for 4-QAM and 16-QAM.}
\label{fig:asep_unified}
    \end{subfigure}
  ~ 
    \begin{subfigure}[t]{0.31\textwidth}
       \centerline{\includegraphics[width=  1.9in]{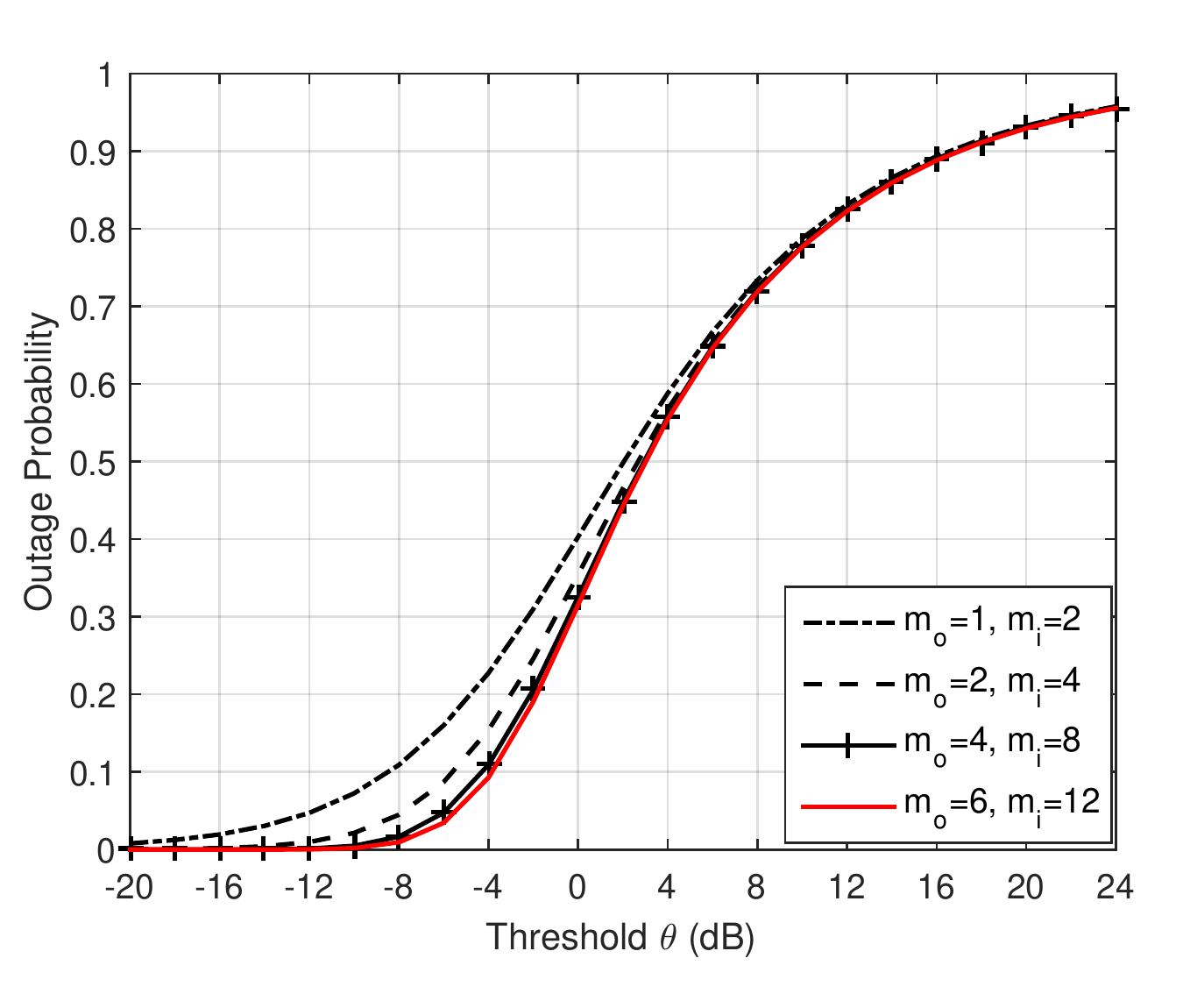}}\caption{\,Outage Probability for $\frac{m_o}{m_i}=\frac{1}{2}$ at $\theta=5$ dB.}
\label{fig:ch_harden_effect}
    \end{subfigure}

    \caption{Unified performance versus the ratio $\frac{m_o}{m_i}$ for an arbitrary MIMO setup.}
    \label{combined}
\end{figure*}

\subsection{{Diversity-Multiplexing Tradeoffs \& Design Guidelines}}

For a fixed number of antennas $N_t=2$ and $N_r=2$, Fig. \ref{fig:comparison_Nt2_Nr2} and Table \ref{tab:rates} compare the performance of the considered MIMO configurations in terms of error probability, outage probability, ergodic rate, and throughput\footnote{ The throughput is defined  as the number of successfully transmitted bits pcu and is given by $\log_2(M) (1-{\rm ASEP})$.}, and quantify the achievable gains with respect to the SISO configuration. Note that, for SIMO and MISO, $N_t$ and $N_r$ are set to $1$, respectively.   Further, in SDMA scenario, the number of single-antenna users served in the network is $\mathcal{K}=2$.  Fig.~\ref{fig:comparison_Nt2_Nr2} and Table \ref{tab:rates} clearly show the diversity-multiplexing tradeoff in cellular networks. The figure shows the outage probability improvement due to diversity, in which the OSTBC achieves the highest outage probability reduction. This is because OSTBC provides both transmit and receive diversity while MISO and SIMO provide either transmit or receive diversity.  Note that despite that MISO and SIMO have the same performance, the SIMO is preferred because it relies on the receive CSI which is easier to obtain than the transmit CSI. {The figure also shows the negative impact of multiplexing on the per-stream ASEP and outage probability in ZF-Rx and MU-MIMO schemes. However, multiplexing several streams per BS improves the overall ergodic rate and per-cell throughput as shown in Table \ref{tab:rates}. }

\begin{table}[h!]
\begin{center}
\begin{tabular}{|c| c| c| c | c c|  }

\hline 

 \textbf{MIMO Setup} &  $\boldsymbol{m_o}$  & $\boldsymbol{m_i}$ & Ergodic Rate & \multicolumn{2}{ |c| } {No. of bits pcu } \\ 

 & & &  (bits/sec/Hz)&  $4$-QAM & $16$-QAM  \\
\hline \hline 

 SIMO & $2$ & $1$ & 2.9523  & $ 1.6926$ &  $2.1772$
 
  \\  \hline

  OSTBC &  $4$ & $2$   & $  {2.9771} $  &  $ {1.7228}$  & $  {2.2044}$
 \\  
\hline 
  ZF-Rx &  $1$ & $2$ &  $3.1644 $ & $ 2.6300$   &  $2.7008$
 \\  
\hline
 SDMA &   $1$ & $2$     & $3.1644 $  & $    2.6300 $ &  $2.7008$
\\ 

\hline
MISO &   $2$ & $1$     & $ 2.9523 $  & $   1.6926$ & $2.1772$
\\ 
\hline

SISO &   $1$ & $1$     & $1.48899 $  & $   1.4780$ & $1.6936$
\\ 
 
\hline
\end{tabular}

\end{center}
\caption{Overall achievable and actual rates gains per cell, with respect to SISO networks, for the different MIMO setups, in an interference-limited scenario for $M=4,16$-QAM modulation scheme.}
\label{tab:rates}
\end{table}

{The results in Fig.~\ref{fig:comparison_Nt2_Nr2} and Table \ref{tab:rates} show the diversity-multiplexing tradeoffs that can be achieved for a $2 \times 2$ MIMO setting. However, as $N_t$ and $N_r$ grow, several diversity and multiplexing tradeoffs are no longer straightforward to compare. Hence, it is beneficial to have a unified methodology to select the appropriate diversity, multiplexing, and number of antennas to meet a certain design objective. From Proposition~\ref{prop} and the subsequent results we noticed two important insights:}
\begin{itemize}
\item The performances of MIMO schemes differ according to their relative $m_o$ and $m_i$ values. In other words, MIMO configurations with equal  $\frac{m_o}{m_i}$ ratio have equivalent per-stream performance. 
\item Multiplexing more data streams increases $m_i$ and does not affect $m_o$. On the other hand, diversity increases $m_o$ and does not affect $m_i$. In other words, $m_o$ represents the diversity gain and $m_i$ represents the number of independently multiplexed data streams per BS (i.e., multiplexing gain).
\end{itemize}
Based on there aforementioned insights, we plot the unified MIMO outage probability and ASEP performance results in Fig.~\ref{combined}. The Figs.~\ref{fig:out_unified} and \ref{fig:asep_unified} show the ASEP and outage probability for a varying ratio of $\frac{m_o}{m_i}$  which can be used for all considered MIMO schemes. Conversely, Fig.~\ref{combined} presents a unified design methodology for MIMO cellular networks as shown in Fig.~\ref{fig:flowchart}. Such unified design provides reliable guidelines for network designers and defines the different flavors of the considered MIMO configurations in terms of achievable diversity and/or multiplexing gains. For instance, for an ASEP or outage probability constraint, the corresponding ratio $\frac{m_o}{m_i}$ and modulation scheme are determined. Then, the network designer can determine the MIMO technique depending on the number of data streams (or number of users) that need to be simultaneously served (i.e., determine $L$ or $\mathcal{K}$). Finally, the number of transmit and receive antennas for the selected MIMO scheme can be determined from Table~\ref{tab:summary}. 

Figs. \ref{fig:out_unified} and \ref{fig:asep_unified} clearly show that incrementing the ratio $\frac{m_o}{m_i}$ enhances the diversity gain whereas decrementing it provides a higher multiplexing gain. That is, network designers are able to maintain the same per-stream ASEP/outage probability by appropriately adjusting the operational parameters, namely, $N_r$, $N_t$ and $L$ (or $\mathcal{K}$ for SDMA). This is done by compensating $m_i$ with the adequate $m_o$ such that  $\frac{m_o}{m_i}$  is kept constant.  For instance, consider a network that needs to increase the number of served users $\mathcal{K}$ without compromising the reliability performance  of each served user.  According to Table~\ref{tab:summary} and Fig.~\ref{combined},  this is achieved by keeping $\frac{m_o}{m_i} = c$, where $c$ is a constant, which hence costs the network additional $\left \lceil \mathcal{K}(c+1)-1  \right \rceil $ transmit antennas per BS.

It is worth mentioning that a design based on the ASEP is more tangible as it is sensitive to the used constellation, as opposed to the outage probability as shown in Fig.~\ref{fig:comparison_Nt2_Nr2}, and Table~\ref{tab:rates}. 
It is also important to note that increasing $m_o$ for a fixed  $\frac{m_o}{m_i}$ ratio can slightly vary the outage probability due to the channel hardening effect as shown in Fig~\ref{fig:ch_harden_effect}. However, such variation is shown to be negligible for $m_o>2$. Another noteworthy observation is that the second term in \eqref{eq:asep_general}, which corresponds to the  $\text{erfc}^2 (\cdot)$ term in \eqref{assep}, requires threefold nested integrals that involve hypergeometric functions to evaluate the ASEP. Such integration is computationally complex to evaluate and may impose some numerical instability specially for large arguments of $m_o$ and $m_i$. In order to overcome such complexity and numerical instability, we invoke Jensen's inequality to the  $\text{erfc}^2 (\cdot)$  term in \eqref{assep}. Hence, the ASEP function becomes $\text{ASEP}\left(\Upsilon\right) \approx {w_{1} }  \mathbb{E}\left[\text{erfc}\left(\sqrt{\beta \Upsilon } \right)\right] + {w_{2} }  \mathbb{E}\left[\text{erfc} \left(\sqrt{\beta \Upsilon} \right)\right]^2 $, which reduces one integral from the second term of \eqref{eq:asep_general}. Using Jensen's inequality yields a stable and accurate approximation compared to \eqref{eq:asep_general} as shown in Figure \ref{fig:asep_unified}, where the red curves represent the numerically unstable ASEP performance as arguments grow, while the black curves represent the Jensen's inequality bounds.

\begin{figure}[t]
 \begin{center} 
{\includegraphics[width=0.75 \columnwidth]{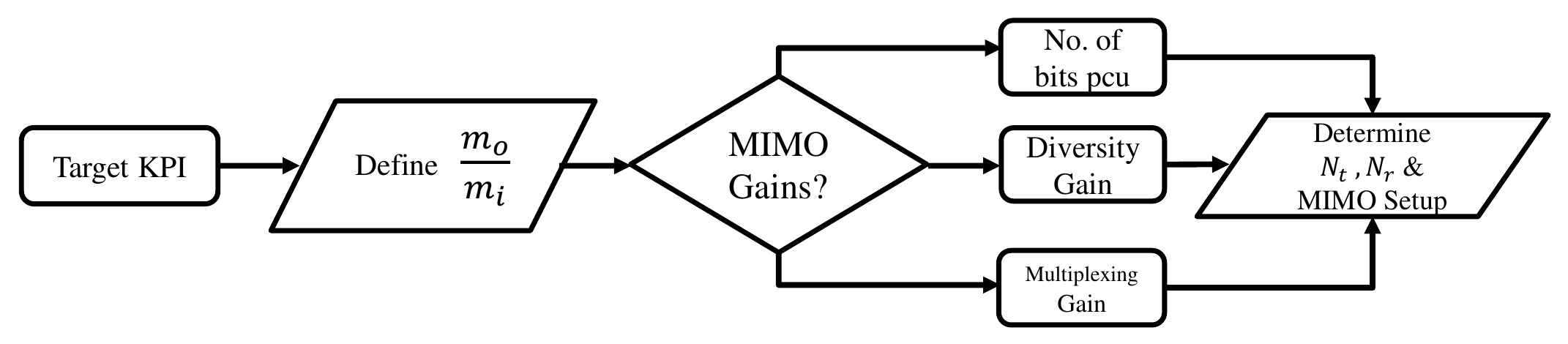}}
\caption{{Flowchart for MIMO selection in cellular networks}}
 \label{fig:flowchart}
 \end{center}
\end{figure}





\section{Conclusion}
\label{sec:conc}
{This paper provides a unified tractable framework for studying symbol error probability, outage probability, ergodic rate, and throughput for downlink cellular networks with different MIMO configurations. The developed model also captures the effect of temporal interference correlation on the outage probability after signal retransmission. The unified analysis is achieved by Gaussian signaling approximation along with an equivalent SISO-SINR representation for the considered MIMO schemes. The accuracy of the proposed model is verified against Monte-Carlo simulations. To this end, we shed lights on the diversity loss due to temporal interference correlation and discuss the diversity-multiplexing tradeoff imposed by MIMO configurations. Finally, we propose a unified design methodology to choose the appropriate diversity, multiplexing, and number of antennas to meet a certain design objective.} 

\appendices

\section{Proof of Theorem~\ref{theo:unified}}
\label{app:outage}

The ASEP expression in \eqref{eq:asep_general} is obtained by taking the expectation over $\Upsilon$ and then using expressions from \cite[eq. (11), (21)]{unif_model} as has been detailed in\cite{Laila_letter}.

For the outage probability, conditioned on $r_o$,
\small{\begin{align}
\!\!\!\! \!\!\!\! \mathcal{O}\left(r_o, \theta\right) &=\mathbb{E}\left[  \mathbb{P}\left( g_o <  \frac{  \theta \mathcal{I} }{P r_o^{-\eta} }  \right)\right] \overset{(d)}{=} \mathbb{E}  \left[ 1- \sum_{j=0}^{m_o-1} \frac{1}{j !} \left( \frac{  \theta \mathcal{I} }{ P r_o^{-\eta} } \right)^j \exp{\left\{ \frac{ - \theta \mathcal{I} }{ P r_o^{-\eta} } \right\}} \right], 
\end{align}
}\normalsize
 where $(d)$ follows from the CDF of the gamma distribution, and then \eqref{out} is obtained from the rules of differentiation of the LT, together with averaging over the PDF of $r_o$.

Ergodic rate expression in \eqref{rate} follows from \cite[Lemma~1]{khairi_useful_lemma}, and by exploiting the independence between the useful and interfering signals, as well as incorporating the CDF of the gamma random variable.

\section{Proof of Lemma \ref{joint_LT}}
\label{app:joint_LT1}

Let $\tilde{\Psi}_{1}^o \subset \tilde{\Psi}^o$ and $\tilde{\Psi}_{2}^o \subset \tilde{\Psi}^o$ be the sets of interfering BSs in the first and second time slots of transmissions, respectively.  Exploiting the independent transmission assumption per time slot,  $\tilde{\Psi}_{1}^o $ and $\tilde{\Psi}_{2}^o $ can be decomposed into the three independent PPPs $\{\tilde{\Psi}_{1}^o \setminus \tilde{\Psi}_{2}^o \}$,  $\{\tilde{\Psi}_{2}^o \setminus \tilde{\Psi}_{1}^o \}$, and $\{ \tilde{\Psi}_{2}^o \cap \tilde{\Psi}_{1}^o \} $ with intensities $p (1-p)  \lambda_B$, $(1-p)p  \lambda_B$, and  $p^2 \lambda_B$, respectively. Substituting $p \lambda_B = \lambda$, the joint LT of the two random variables $\mathcal{I}_1$, and $\mathcal{I}_2$ is derived as follows,

\footnotesize{
\begin{align}
&\mathcal{L}_{\mathcal{I}_1, \mathcal{I}_2  \vert r_o}\left(z_1,z_2 \right) 
=\mathbb{E} \left[\exp{\left\{ - \displaystyle \sum_{r_i \in \tilde{\Psi}_{2}^o \cap \tilde{\Psi}_{1}^o }  P r_i^{-\eta} \left( z_1 \tilde{g}_i^{(1)} + z_2 \tilde{g}_i^{(2)}\right)     - z_1  \displaystyle \sum_{r_i \in \tilde{\Psi}_1^o \setminus  \tilde{\Psi}_2^o } P r_i^{-\eta} \tilde{g}_i^{(1)} - z_2  \displaystyle \sum_{r_i \in \tilde{\Psi}_2^o \setminus  \tilde{\Psi}_1^o } P r_i^{-\eta} \tilde{g}_i^{(2)}\right\} } \right] \notag \\
\!\!\!\!& \overset{(e)}{=} \exp\left\{ -2 \pi p \lambda \int_{r_o}^{\infty} \!\!\!\! \mathbb{E}_{\tilde{g}_i^{(1)},\tilde{g}_i^{(2)}} \left[1- e^{- P x^{-\eta }\left(z_1 \tilde{g}_i^{(1)} + z_2 \tilde{g}_i^{(2)}\right) } \right] x \mathrm{d}x  -   2 \pi (1- p) \lambda \int_{r_o}^{\infty} \!\!\!\! \mathbb{E}_{\tilde{g}_i^{(1)}} \left[1- e^{- P x^{-\eta } z_1 \tilde{g}_i^{(1)} } \right] x \mathrm{d}x\right.  \notag \\
&\left. \quad \quad\quad\quad\quad-   2 \pi (1- p) \lambda \int_{r_o}^{\infty} \!\!\!\! \mathbb{E}_{\tilde{g}_i^{(2)}} \left[1- e^{- P x^{-\eta } z_2 \tilde{g}_i^{(2)} } \right] x \mathrm{d}x\right\} \notag   \\
&  \overset{(f)}{=}  \exp\left\{ -2 \pi  \lambda\int_{r_o}^{\infty} \left( p \left[1- {\frac{1}{\left(1+P x^{-\eta } z_1\right)^{m_{i,1}}}\frac{1}{\left(1+P x^{-\eta } z_2\right)^{m_{i,2}}} } \right]+ (1-  p) \left[1- {\frac{1}{\left(1+P x^{-\eta } z_1\right)^{m_{i,1}}} } \right] \right.\right.  \notag \\
&\left. \left. \quad \quad\quad\quad\quad\quad\quad\quad\quad+ (1-  p) \left[1- {\frac{1}{\left(1+P x^{-\eta } z_2\right)^{m_{i,2}}} } \right]\right) x \mathrm{d}x\right\},
\end{align}
}\normalsize
where $(e)$ is obtained from the PGFL and exploiting the independence between the PPPs $\{\tilde{\Psi}_{1}^o \setminus \tilde{\Psi}_{2}^o \}$,  $\{\tilde{\Psi}_{2}^o \setminus \tilde{\Psi}_{1}^o \}$, and $\{ \tilde{\Psi}_{2}^o \cap \tilde{\Psi}_{1}^o \} $ \cite{martin_book}, and $(f)$ follows from the LT of the two independent gamma distributed random variables $\tilde{g}^{(1)}_i$ and $\tilde{g}^{(2)}_i$. Solving the integral completes the proof.

\section{Proof of Theorem \ref{theo:joint_LT}}
\label{app:joint_LT2}
The joint CCDF of $\bar{\Upsilon}_1$ and $\bar{\Upsilon}_2$ is given by
\footnotesize{
\begin{align}
\!\!\!\!\!\!\!\!\!\!\!\!\!\!\mathbb{P}\left( \bar{\Upsilon}_1 > \theta, \bar{\Upsilon}_2 > \theta\right)&=\mathbb{E}  \left[  \mathbb{P}\left( \tilde{g}_o^{(1)} >  \frac{  \theta \mathcal{I}_1 }{P r_o^{-\eta} }  , \:\:\:  \tilde{g}_o^{(2)} >  \frac{  \theta \mathcal{I}_2 }{P r_o^{-\eta} }   \right)\right] \notag \\
&\overset{(i)}{=} \mathbb{E}  \left[ \sum_{j_1=0}^{m_{o,1}-1}  \sum_{j_2=0}^{m_{o,2}-1} \frac{1}{j_1 ! j_2 !} \left( \frac{ - \theta }{ P r_o^{-\eta} } \right)^{j_1+j_2} \left( \mathcal{I}_1^{j_1} \mathcal{I}_2^{j_2} \right) \exp{\left\{ \frac{ - \theta \left(\mathcal{I}_1 + \mathcal{I}_2\right) }{ P r_o^{-\eta} } \right\}}   \right] \notag\\
&\overset{(ii)}{=} \mathbb{E}  \left[ \sum_{j_1=0}^{m_{o,1}-1}  \sum_{j_2=0}^{m_{o,2}-1} \frac{1}{j_1 ! j_2 !} \left( \frac{ - \theta }{ P r_o^{-\eta} } \right)^{j_1+j_2} \frac{\partial^{(j_1+j_2)}}{\partial z_1^{j_1} \partial z_2^{j_2}} \mathcal{L}_{\mathcal{I}_1, \mathcal{I}_2 \vert r_o}\left(z_1,z_2 \right) \bigg \vert_{z_1= z_2=\frac{\theta}{P r_o^{-\eta}}} \right], 
\end{align}
}\normalsize
such that $(i)$ follows from the independence of $ \tilde{g}_o^{(1)}$ and $\tilde{g}_o^{(2)}$ along with the CCDF of their Gamma distributions. $(ii)$ is obtained by utilizing the LT identity $ t_1^a t_2^b  f(t_1,t_2)\Leftrightarrow \frac{\partial^{(a+b)}}{\partial z_1^a \partial z_2^b } \mathcal{L}_{t_1,t_2}(z_1,z_2)$ , which can be proved as follows. First, we write the joint Laplace Transform of two variables $t_1$ and $t_2$ as
\begin{align}
\mathcal{L}_{t_1,t_2}\left(z_1,z_2\right) =\int_0^{\infty} \int_0^{\infty} {f(t_1,t_2)} e^{-z_1 t_1} e^{-z_2 t_2} \mathrm{d}t_1 \mathrm{d}t_2,
\end{align}
then,
\footnotesize{
\begin{align}
\frac{\partial^{j_1 + j_2}\mathcal{L}_{t_1,t_2}\left(z_1,z_2\right)}{\partial z_1^{j_1}\partial z_2^{j_2}}& = \frac{\partial^{j_1 + j_2}}{\partial z_1^{j_1}\partial z_2^{j_2}} \int_0^{\infty} \int_0^{\infty} {f(t_1,t_2)} e^{-z_1 t_1} e^{-z_2 t_2} \mathrm{d}t_1 \mathrm{d}t_2 \notag \\
& = \int_0^{\infty} \int_0^{\infty} \left( -1\right)^{j_1 + j_2}\left(  t_1 \right)^{j_1} \left(  t_2 \right)^{j_2} {f(t_1,t_2)} e^{-z_1 t_1} e^{-z_2 t_2} \mathrm{d}t_1 \mathrm{d}t_2,  
\end{align}
\normalsize
where the second equality follows  by Leibniz rule and applying the rules of partial differentiation, which proves the identity.





\section{}
\label{app:all_lemmas}
\subsection{Proof of Lemma~\ref{lemma:SIMO}}
\label{app:SIMO}
In SIMO transmission, by applying MRC at the receiver side, for $\bar{\textbf{w}}_o^T = \textbf{h}_o^{H}$, then the post-processed signal is written as
\small{\begin{align}
\tilde{{y}}&=\bar{\textbf{w}}_o^T \textbf{y}=  \sqrt{P} r_o^{-\frac{\eta}{2}} \| \textbf{h}_o\|^2  \:  {s}_o  +\sum_{r_i \in \tilde{\Psi}^o} \sqrt{P} r_i^{-\frac{\eta}{2}} \textbf{h}_o^{H} \textbf{h}_i  \tilde{s}_i   +  \textbf{h}_o^{H} \textbf{n}.
\label{eq:simo_eq}
\end{align}}\normalsize
We start with computing the effective noise variance since a post-processor is applied. The noise power is expressed as
\small{\begin{align}
\text{Var}_{\textbf{n}}\left[\textbf{h}_o^{H}\textbf{n}\right]=\mathcal{N}_o  \| \textbf{h}_o\|^2.
\end{align}}\normalsize
Therefore, the random variable $\epsilon=  \| \textbf{h}_o\|^2$, is used to normalize the resultant interference power. The effective interference variance conditioned on the network geometry and the intended channel gains with respect to $\tilde{s}_i$ is given by
\small{
\begin{align}
\mathcal{I}=\frac{1}{\epsilon}\text{Var}_{\tilde{s}_i}\left[\sum_{r_i \in \tilde{\Psi}^o} \sqrt{P}  r_i^{-\frac{\eta}{2}}  \:\textbf{h}_o^{H}\textbf{h}_i  \tilde{s}_i \right]=\sum_{r_i \in \tilde{\Psi}^o} P r_i^{-\eta} \frac{\vert \textbf{h}_o^{H} \textbf{h}_i\vert^2  }{  \| \textbf{h}_o\|^2 }.
\end{align} 
}\normalsize
By inspection of the interference variance, it is clear that $\left\{\breve{\textbf{{H}}}\right\}=\left\{\textbf{H}_o \right\}$. Also, we notice that there exists only one coefficient $ a^{(i)}_{l}= \frac{ \textbf{h}_o^{H}  \textbf{h}_i }{  \| \textbf{h}_o\| }$. Recall that the number of independent coefficients $ a^{(i)}_{l,k}$ depends on the number of independent transmitted streams, which is equal to one in the SIMO case.  Accordingly, $\tilde{g}_i= \big \vert  a^{(i)}_{l} \big \vert^2   \sim \text{Gamma}\left( m_i,1\right)$, with $m_i=1$. Similarly, conditioned on the intended and interfering channel gains, the received signal power, with respect to the transmitted signal, can be shown to be
\begin{align}
\mathcal{S}=\frac{1}{\epsilon}\text{Var}_{s_o}\left[ \sqrt{P} r_o^{-\frac{\eta}{2}} \| \textbf{h}_o\|^2   {s}_o \right]=P r_o^{-\eta}   \| \textbf{h}_o\|^2.
\end{align} 
Therefore, $\tilde{g}_o = \| \textbf{h}_o \|^{2} \sim \text{Gamma}\left( m_o,1\right)$ where $m_o=N_r$.

\subsection{Proof of Lemma~\ref{lemma:OSTBC}}
\label{app:OSTBC}
Employing OSTBC, the received vector at a typical user at time instant $\tau, N_t \leq\tau\leq {T}$, is given by
\small{\begin{align}
\textbf{y}\left( \tau \right)= \sqrt{P}r_o^{-\frac{\eta}{2}}{\textbf{H}_{o}} \textbf{s}+ \sum_{r_i \in \tilde{\Psi}^o} \sqrt{P} r_i^{-\frac{\eta}{2}} \textbf{H}_i \tilde{\textbf{s}}_i+ \textbf{n}\left(\tau \right).
\end{align}} \normalsize
 Let $\boldsymbol{\mathcal{Y}}$ be the stacked vector of received symbols over ${T}$ intervals, and let $L=N_t$, such that,
\small{\begin{align}
\boldsymbol{\mathcal{Y}} =  \sqrt{P} r_o^{-\frac{\eta}{2}}\textbf{H}_{\text{eff}}\: \textbf{s} +  \textbf{i}_{\text{agg}} + \textbf{n}.
\end{align} \normalsize
\noindent where $\boldsymbol{\mathcal{Y}} \in \mathbb{C}^{{T}\cdot N_r \times 1}$, and $ \textbf{i}_{\text{agg}}$ is the concatinated aggregate interference ${T}\cdot N_r \times 1$ vector. The effective channel matrix $\textbf{H}_{\text{eff}} \in \mathbb{C}^{{T}\cdot N_r \times N_t}$ is expressed as a linear combination of the set of dispersion matrices $\mathcal{A}$ and $\mathcal{B}$ chosen according to the adopted orthogonal space-time code as follows \cite{dispersion_babak_hassibi,larsson_book},
\small{\begin{align}
\label{eq:dispersion_matrix}
\textbf{H}_{\text{eff}}=\sum_{j=1} ^{N_r}\sum_{q=1}^{N_t} \alpha_{jq} \mathcal{A}_{jq}+ \jmath \beta_{jq} \mathcal{B}_{jq},
\end{align}\normalsize
\noindent where  $h_{jq}=\alpha_{jq} + \jmath \beta_{jq}$. Moreover,  $ \textbf{H}_{\text{eff}}^{H}\textbf{H}_{\text{eff}}=\| \textbf{H}_o\|^{2}_{\text{F}} \textbf{I}$ where $\| \textbf{H}_o\|^{2}_{\text{F}}=\sum_{j=1}^{N_r}\sum_{q=1}^{N_t} \vert h_{jq}\vert^2$  is the squared Frobenius norm of the intended channel matrix. Hence,  $\| \textbf{H}_o\|^{2}_{\text{F}} \sim \frac{1}{2} \chi^2 \left(2 N_s N_r \right)$ \cite{paulraj_book}.
Moreover, the aggregate interfering signals are expressed as
\small{\begin{align}
\textbf{i}_{\text{agg}}
= \displaystyle \sum_{r_i \in \tilde{\Psi}^o} \sqrt{P}r_i^{-\frac{\eta}{2}} \textbf{H}_{i,\text{eff}}\:  \tilde{\textbf{s}}_i,
\end{align}}\normalsize
\noindent such that $\textbf{H}_{i,\text{eff}}$ is defined similar to \eqref{eq:dispersion_matrix}. For detection, we equalize the effective channel matrix at the receiver side by  $\textbf{W}_o$. Hence, the received vector $\boldsymbol{\tilde{\mathcal{Y}}}$ is written as
\small{\begin{align}
\!\!\!\!\boldsymbol{\tilde{\mathcal{Y}}}&=\textbf{W}_o \boldsymbol{\mathcal{Y}} =  \sqrt{P}r_o^{-\frac{\eta}{2}}\| \textbf{H}_o\|_{\text{F}} \: \textbf{s} +\sum_{r_i \in \tilde{\Psi}^o} \sqrt{P}r_i^{-\frac{\eta}{2}}   \textbf{A}_i \tilde{\textbf{s}}_i +  \textbf{w},
\end{align}
 \normalsize 
\noindent such that $\textbf{w}=\textbf{W}_o  \textbf{n}$ and  $\textbf{A}_i=\textbf{W}_o  \textbf{H}_{i,\text{eff}}$ with elements $a^{(i)}_{l,k}$ as defined in \eqref{eq:mimo_rxed_signal}.
Without loss of generality, let us consider the detection of the ${l^{th}}$ arbitrary symbol from the received vector $\boldsymbol{\tilde{\mathcal{Y}}}$. Due to the adopted Gaussian signaling scheme, we lump interfernce with noise, and thus it is essential to obtain the interference variance. First, let us define $\boldsymbol{q}_k$ as the $k^{th}$ column of the matrix  $\textbf{H}_{\text{eff}}$, similarly, $\boldsymbol{q}_{i,k}$ is the $k^{th}$ column of the matrix $\textbf{H}_{i,\text{eff}} $. Then, the received interference variance for the $l^{th}$ symbol denoted as $\mathcal{I}_l$, computed with respect to the interfering symbols $\tilde{\textbf{s}}_i$ can be derived as
\small{
\begin{align}
\!\!\!\mathcal{I}_l&=\text{Var}_{\textbf{s}_{i}}\left[ \sum_{r_i \in \tilde{\Psi}^o} \sum_{k=1}^{N_s} \sqrt{P}r_i^{-\frac{\eta}{2}} \frac{\boldsymbol{q}_l^{H} \boldsymbol{q}_{i,k}}{\| \textbf{H}_o\|_{\textbf{F}}} s_{i,k}  \right]= \sum_{r_i \in \tilde{\Psi}^o}   \sum_{k=1}^{ N_s}  {P}r_i^{-\eta}  \frac{\big \vert\boldsymbol{q}_l^{H} \boldsymbol{q}_{i,k}\big \vert^2  }{\| \textbf{H}_o\|^2_{\textbf{F}}},
\end{align}
} \normalsize
\noindent where the summation is over the $N_s$ active antennas per transmission. Note that, conditioned on $\left\{\breve{\textbf{{H}}}\right\}=\left\{\textbf{H}_o\right\}$, $a^{(i)}_{l,k} =\frac{\boldsymbol{q}_l^{H} \boldsymbol{q}_{i,k}}{\| \textbf{H}_o\|_{\textbf{F}}}$ is a normalized and independently weighted sum of complex Gaussian random variables, thus $a^{(i)}_{l,k} \sim \mathcal{CN} \left(0,1\right)$. Although a post-processor is applied, the noise power is maintained to be $\mathcal{N}_o$.
Thus, $\tilde{g}_i \sim \left(m_i,\Omega_i \right)$ with $m_i=N_s$ and $\Omega_i=1$.
Similarly, the received signal power is found to be 
\vspace{-0.3cm}
\begin{align}
\mathcal{S}&= \text{Var}_{\textbf{s}}\left[ \sqrt{P}r_o^{-\frac{\eta}{2}} \| \textbf{H}_o\|_{\text{F}}\right]= P r_o^{-\eta} \| \textbf{H}_o\|^{2}_{\text{F}},
\end{align}
where $\tilde{g}_o \sim \text{Gamma}\left(m_o, \Omega_o \right)$ with $m_o= N_s N_r$ and $\Omega_o=1$. 



\subsection{Proof of Lemma~\ref{lemma:ZF_Rx}}
\label{app:ZF_Rx}
Without loss of generality, we focus on the detection of an arbitrary symbol $l$ from the received vector $\tilde{\textbf{y}}=\textbf{W}_o \textbf{y}$, given by
\begin{align}
\tilde{y}_l= \sqrt{P} r_o^{-\frac{\eta}{2}}  \:  {s}_{l}  +\sum_{r_i \in \tilde{\Psi}^o} \sqrt{P} r_i^{-\frac{\eta}{2}} \bar{\textbf{w}}^T_{o,l} \textbf{H}_i \tilde{\textbf{s}}_i   +\bar{\textbf{w}}^T_{o,l} \textbf{n},
\end{align}

\noindent which is similar to \eqref{eq:mimo_rxed_signal}. 
First we need to to obtain the received noise variance since a post-processing matrix is applied and thus the noise variance is scaled. Conditioned on $\textbf{H}_o$, the received noise power is defined as
\begin{align}
\text{Var}_{\textbf{n} }\left[ \bar{\textbf{w}}^T_{o,l} \textbf{n}  \right]&=  \bar{\textbf{w}}^T_{o,l}\mathbb{E}\left[ \textbf{n}  \textbf{n}^H  \right] \bar{\textbf{w}}^*_{o,l} = \mathcal{N}_o  \left(\textbf{W}_o \textbf{W}_o^H\right)_{ ll} = \mathcal{N}_o  \left(\textbf{H}_o \textbf{H}_o^H\right)^{-1}_{ ll}.
\end{align}
Then, the scaling random variable is $\epsilon= \left(\textbf{H}_o \textbf{H}_o^H\right)^{-1}_{ ll}$. Next, we obtain the effective interference variance from the $l^{th}$ received symbol as
\small{
\begin{align}
\!\!\! \!\!\! \mathcal{I}_l&= \frac{1}{\epsilon}\text{Var}_{\tilde{\textbf{s}}_i}\left[   \sum_{r_i \in \tilde{\Psi}^o} \sqrt{P} r_i^{-\frac{\eta}{2}} \bar{\textbf{w}}^T_{o,l} \textbf{H}_i \tilde{\textbf{s}}_i \right]
= \frac{1}{\epsilon} \sum_{r_i \in \tilde{\Psi}^o} P r_i^{-\eta}  \left(\textbf{W}_o \textbf{W}_o^H\right)_{ ll}  \left(\textbf{H}_i \textbf{H}_i^H\right)_{ ll}
=  \sum_{r_i \in \tilde{\Psi}^o} P r_i^{-\eta} \left(\textbf{H}_i \textbf{H}_i^H\right)_{ ll}.
\end{align}
 \normalsize
In this case, the processing resulting interference channel set $\left\{\breve{\textbf{{H}}}\right\}=\emptyset$. Therefore, $a^{(i)}_{l,k} = \left(\textbf{H}_i \textbf{H}_i^H\right)_{ ll}$ and $\tilde{g}_i \sim \left( m_i,\Omega_i\right)$, with $m_i=N_t$ and $\Omega_i=1$. The received signal power is similarly computed as
\small{ \begin{align}
\!\!\! \!\!\! \mathcal{S}=\frac{1}{\epsilon}\text{Var}_{\textbf{s}}\left[   \sqrt{P} r_o^{-\frac{\eta}{2}} {s}_l \right] = \frac{  P r_o^{-\eta} }{ \left(\textbf{H}_o \textbf{H}_o^H\right)^{-1}_{ ll} }.
\end{align}} \normalsize

\vspace{-0.3cm}
\noindent Since $g_w=\left(\textbf{H}_o \textbf{H}_o^H\right)^{-1}_{ ll} \sim \text{Inv-Gamma}\left(N_r-N_t+1,1\right)$ \cite{paulraj_book}. 
Then, we can let $\frac{1}{g_w}=\tilde{g}_o \sim \text{Gamma} \left( m_o, \Omega_o\right)$, where $m_o=N_r-N_t+1$ and $\Omega_o=1$.

\subsection{Proof of Lemma~\ref{lemma:SDMA}}
\label{app:SDMA}
In a multi-user MIMO setting, we introduce a slight abuse of notation for the intended and interfering channel matrices such that they are of dimensions $\mathcal{K} \times N_t$. 
The received interference power at user $l$ where $1\leq  l \leq \mathcal{K}$, averaged over the interfering symbols $\boldsymbol{\tilde{s}_i}$ is given by
\small{
\begin{align}
\mathcal{I}_{ l}&=\text{Var}_{\boldsymbol{\tilde{s}_i}}\left[   \sum_{r_i \in \tilde{\Psi}^o} \sqrt{P} r_i^{-\frac{\eta}{2}} \textbf{h}_{i,  l} \textbf{V}_i \boldsymbol{\tilde{s}_i} \right] 
= \sum_{r_i \in \tilde{\Psi}^o} P r_i^{-\eta} \| \textbf{h}_{i,  l}  \textbf{V}_i \|^2, 
\end{align}
} \normalsize
\noindent where $\textbf{h}_{i,l}$ is the $l^{th}$ row of $\textbf{H}_i$ and $\left\{\breve{\textbf{{H}}}\right\}=\left\{{\textbf{V}}_i \right\}$. Also, $\| \textbf{h}_{i,  l}  \textbf{V}_i \|^2 = \sum_{l=1}^{\mathcal{K}} \vert a^{(i)}_{l,k}  \vert^2 $. However, the column vectors of $\textbf{V}_i$  are \emph{not} independent. Therefore, conditioned on $\textbf{v}_{i, l}$, the linear combination $ \sum_{ l=1}^{\mathcal{K}}\big \vert a^{(i)}_{l,k} \vert^2$ does not follow a Gamma distribution. Nevertheless, for tractability we \emph{approximate} this summation by a Gamma distribution. Thus, $\tilde{g}_i \sim \text{Gamma}\left( m_i, \Omega_i \right)$, where $m_i=\mathcal{K}$ and $\Omega_i=1$ by assuming such independence. This renders the aggregate interference power distribution at user $ l$ an \emph{approximation}. Similarly, the useful signal power at user $ l$ is straightforward to be obtained, after appropriate diagonalization, as $\mathcal{S}=\frac{1}{\| \textbf{v}_{ l} \|^2} \sim {\text{Gamma}\left(m_o,\Omega_o\right)}$, with $m_o= N_t - \mathcal{K}+1$ and $\Omega_o=1$ \cite{mimo_comm_book}. This can also be interpreted as having the precoding matrix nulling out $\mathcal{K}-1$ directions out of the $N_t$ subspace at the transmitter side. 

\subsection{Proof of Lemma~\ref{lemma:SM}}
\label{app:SM}

\noindent Since, there are $N_t$ distinct multiplexed symbols to be transmitted, we will study the pairwise error probability (PEP) of two distinct transmitted codewords, denoted as $\mathcal{P}=\mathbb{P}\left\{ \tilde{\textbf{s}}=\textbf{s}_1 \vert \textbf{s}_o\right\}$.
For ease of notaion, let $\delta=\sqrt{P}r_o^{-\frac{\eta}{2}}$. Therefore,
\begin{equation}
 \| \textbf{y}- \delta \textbf{H}_o  \textbf{s}_o  \|^{2} \:\:\:\:  \substack{ \textbf{s}_1\\ > \\ < \\ \textbf{s}_o}   \:\:\:\:  \| \textbf{y}- \delta \textbf{H}_o  \textbf{s}_1  \|^{2}.
\label{PEP_def}
\end{equation}

\noindent Using some mathematical manipulations, and assuming $\textbf{s}_{o}$ was the actual transmitted symbols, it can be shown that 
\small{
\begin{align}
\mathcal{P}\left(\textbf{e}\right)=\mathbb{P}\left\{
  \left[ \textbf{I}_{\text{agg}}^H+ \textbf{n}^H\right] \textbf{H}_o \textbf{e}+ \textbf{e}^H \textbf{H}_o^H \left[  \textbf{I}_{\text{agg}}+ \textbf{n}\right]  >   \delta
  \textbf{e}^H \textbf{H}_o^H  \textbf{H}_o \textbf{e}\right\},
\end{align}
}\normalsize 
\noindent where $ \textbf{I}_{\text{agg}}=\sum_{r_i \in \tilde{\Psi}^o} I_i$.
\noindent Conditioned on the channel matrices $\textbf{H}_o$ and $\textbf{H}_{i}$, and considering the Gaussian signaling approximation, the L.H.S of the above inequality represents the interference-plus-noise power and is a Gaussian random variable, denoted as $\mathcal{V}$ with zero-mean and variance $\sigma^2_{\mathcal{V}}$, thus, $\small{\mathcal{P}\left( \textbf{e}\right)= \frac{1}{2}  \text{erfc}\left( \frac{\delta \: \textbf{e}^H \textbf{H}_o^H \textbf{H}_o\textbf{e} }{\sqrt{2 \sigma^2_{\mathcal{V}}}}  \right),}\normalsize$  where the variance $\sigma^2_{\mathcal{V}}$ is given by
\small{
\begin{align}
\sigma^2_{\mathcal{V}}&= 2  \left[ \:\textbf{e}^H \textbf{H}_o^H \textbf{H}_o \textbf{e}\right] \left(  \mathcal{N}_o+  \sum_{r_i \in \tilde{\Psi}^{o}}  P r_i^{-\eta} \sum_{k=1}^{N_t}  \frac{ \vert \left(\textbf{H}_o\textbf{e}\right)^H \textbf{h}_{i,k} \vert^2}{\| \textbf{H}_o\textbf{e}\|^2_2} \right) .
\end{align}
}\normalsize
By following the same convention used in this paper, it is clear that the interference power is represented as
\vspace{-0.5cm}
\small{
\begin{align}
 \mathcal{I}=\sum_{r_i \in \tilde{\Psi}^{o}}  P r_i^{-\eta} \sum_{k=1}^{N_t}  \frac{ \vert \left(\textbf{H}_o\textbf{e}\right)^H \textbf{h}_{i,k} \vert^2}{\| \textbf{H}_o\textbf{e}\|^2}, 
\end{align}
}\normalsize
and thus we see that $\left\{\breve{\textbf{{H}}}\right\}=\left\{\textbf{H}_o \right\}$ and $a^{(i)}_{l,k} = \frac{ \left(\textbf{H}_o\textbf{e}\right)^H \textbf{h}_{i,k} }{\| \textbf{H}_o\textbf{e}\|}$. Conditioned on $\textbf{H}_o$, $\tilde{g}_i \sim \text{Gamma}\left(N_t,1 \right)$. 
 Furthermore,  let $\mathcal{S} =\textbf{e}^H \textbf{H}_o^H \textbf{H}_o \textbf{e}\overset{d}{=}\| \textbf{e}\|^2 \left( \textbf{H}_o^H \textbf{H}_o\right)_{ll}$, hence it  is starightforward to see that $\left( \textbf{H}_o^H \textbf{H}_o\right)_{ll}$ has a $\chi^2\left(N_r\right)$ distribution. Thus, $\tilde{g}_o = \left( \textbf{H}_o^H \textbf{H}_o\right)_{ll} \sim  \text{Gamma} \left(m_o,\Omega_o \right)$, with $m_o=N_r$ and $\Omega_o=1$. Then, the conditional pairwise error probability is expressed by
\footnotesize{
\begin{align}
\mathcal{P}\left(\textbf{e}\right)&=\frac{1}{2}  \text{erfc}\left( \sqrt{\frac{P r_o^{-\eta} \:  \|\textbf{e}\|^2  \tilde{g}_o }{4 \:  \left( \mathcal{N}_o+   \sum_{r_i \in \tilde{\Psi}^o}P r_i^{-\eta} \tilde{g}_i    \right) }}\right).
\end{align}
}\normalsize

\bibliographystyle{IEEEbib}
\bibliography{IEEEabrv,MyLib}

\end{document}